\documentclass[12pt, DIV=9]{article}
\usepackage[utf8]{inputenc}
\usepackage[english]{babel}
\usepackage{amsmath}
\usepackage{amsthm}
\newtheorem{hypo}{Hypothesis}
\newtheorem{theorem}{Theorem}[section]
\newtheorem{lemma}[theorem]{Lemma}
\newtheorem*{lem}{Lemma}
\newtheorem{prop}[theorem]{Proposition}
\theoremstyle{definition}

\theoremstyle{remark}
\newtheorem{bem}[theorem]{Remark}
\newcommand{\dd}{\mathrm{d}}
\newcommand{\p}{^{\prime}}

\usepackage{color} 
\usepackage{amssymb}
\usepackage{graphicx}
\usepackage{pdflscape}
\usepackage{hyperref}
\usepackage{tensor}
\usepackage{wrapfig}
\usepackage{textcomp}
\usepackage{mathtools}
\usepackage{braket}

\title{Absence of submultiplicative norms for Wick-ordered Operator Products}
\author{Jakob Geisler $<$jakob.geisler@tu-bs.de$>$\\ Institut für Analysis und Algebra\\ TU Braunschweig, Germany}
\date{\today}

\begin{document}

\maketitle

{\noindent \small \textbf{Abstract.} In this paper, a key problem of the rigorous formulation of the renormalization group as a continuous flow is identified. Some essential features of the operator-theoretic renormalization group presented in \cite{sfb} are recalled, and a family of norms associated to different Banach spaces of Hamiltonians is introduced. From this, three different properties, which should be satisfied by an adequate norm, are derived, e.g., a submultiplicativity with respect to Wick-ordering. A proof is given that none of the norms introduced in this paper satisfy all of these conditions.\\ \textbf{MSC:} 81T17, 81S99, 46N50}

\section{Introduction and main result}
The renormalization group (RG), originally introduced by Wilson \cite{Wilson}, is a powerful tool in statistical physics and quantum field theory for analyzing physical systems that can not be treated pertubatively. The main idea is to consider a suitable vector space of Hamiltonians on which one defines a renormalization transformation which preserves important physical properties of the transformed Hamiltonian. Iterating this transformation, physical intuition suggests that in the limit of infinite iterations only some finite-dimensional subspace has to be analyzed.\\
Besides its tremendous success in physics, a rigorous mathematical justification of the discrete RG, especially in the context of non-relativistic QED, has been given for a few physical models (e.g. \cite{haslerherbst, RGAnalysis, sfb}) and still is part of the active research. Different approaches to multiscale analysis, e.g., using Pizzo's method \cite{pizzo} can also be found in \cite{InfraredAlgoQED, InfrarredAlgoQEDII}. In this paper, we identify one of the key problems of the rigorous formulation of the RG as a continuous flow by proving the absence of submultiplicative norms for Wick-ordering.\\

\noindent We start by introducing a Banach space of Hamiltonians, based on \cite{sfb}, but in a slightly more general setting. Let $\mathcal{F}_b(L^2(\mathbb{R}^d))$ be the Fock space of scalar bosons in $d \in \mathbb{N}$ dimensions. In case of massless bosons the energy of the bosonic field is given by
\begin{equation}
    H_{ph} = \int_{\mathbb{R}^d}\limits \dd^dk \, a^\ast(k) \vert k\vert a(k),
\end{equation}
where $(a^\ast(k), a(k))$ are the pointwise bosonic creation- and annihilation-operators, defined as operator-valued distributions, such that
\begin{equation}
    [a(k), a(k^\prime)] = 0 = [a^\ast(k), a^\ast(k^\prime)] \quad, \quad [a(k), a^\ast(k^\prime)] = \delta^d(k-k^\prime).
\end{equation}
Since the groundstate energy $E_{gs} = 0$ of $H_{ph}$ is not isolated from the essential spectrum of $H_{ph}$, naive pertubation theory of operators $H = H_{ph} + W$ with $W \in \mathcal{B}(\mathcal{F}_b(L^2(\mathbb{R}^d)))$ fails. Considering the infrared (and not the ultraviolet) problem, the focus of \cite{sfb} lies in the analysis on the reduced Hilbert space $\mathcal{H}_{red} = P_{red}\mathcal{F}_b(L^2(\mathbb{R}^d))$, where $P_{red} = \mathbf{1}[H_{ph}< 1]$ denotes the projection onto field energies less than one. While we do assume the same restriction onto small photon energies, we stress that this restriction is not important for our analysis. In fact, the projection onto small field energies make the proofs of the absence of submultiplicative norms even harder, and we demonstrate that our proofs apply to the case of an unreduced Hilbert space, too. In \cite{sfb} the groundstate of the Hamiltonian 
\begin{equation}
    H = H_{ph} + \sum_{m+n \geq 1} W_{m,n} \in \mathcal{B}(\mathcal{H}_{red}) \label{eq: specproblem}
\end{equation}
was analyzed, where $W_{m,n}$ is the interaction operator, which creates $m$ bosons and annihilates $n$ bosons. These interactions are parameterized by measurable functions $w_{m,n}: [0,1]\times B_1^{m+n} \mapsto \mathbb{C}$, continuous on $[0,1]$, for almost every $(k_1^m; \Tilde{k}_1^n) \in B_1^{m+n}$ and fulfilling integral bounds specified later, such that
\begin{equation}
    W_{m,n}[w_{m,n}] \coloneqq P_{red} \int_{B_1^{m+n}} \limits \frac{\dd k_1^{m} \dd\Tilde{k}_1^n}{\vert k_1^m\vert^{1/2}\vert \Tilde{k}_1^n\vert^{1/2}} a^\ast(k_1^m) w_{m,n}[H_{ph}; k_1^m; \Tilde{k}_1^n] a(\Tilde{k}_1^n) P_{red} \label{eq: WmnDef}.
\end{equation}
Here, we use the notation
\begin{align}
    k_1^m \coloneqq (k_1, ..., k_m) \in \mathbb{R}^{dm} \quad, \quad \Tilde{k}_1^n \coloneqq (\Tilde{k}_1, ..., \Tilde{k}_n) \in \mathbb{R}^{dn},\\
    \dd k_1^m \coloneqq \prod_{i=1}^m \dd^dk_i \quad, \quad a^\ast(k_1^m)\coloneqq \prod_{i=1}^m a^\ast(k_i)\\
    \vert k_1^m\vert \coloneqq \prod_{i=1}^m \vert k_i\vert \quad, \quad \Sigma[k_1^m] \coloneqq \sum_{i=1}^m \vert k_i\vert\\
    B_1^m \coloneqq B_1 \times B_1\times... \times B_1 \quad, \quad B_1 \coloneqq \{k \in \mathbb{R}^d \vert \vert k\vert \leq 1 \}.
\end{align}
Next, we define a family of norms $\{\Vert \cdot \Vert_{\lambda, p}\}_{\lambda \in \mathbb{R}, p\geq 1}$ that make the space of the functions $w_{m,n}$ a Banach space $\mathcal{W}_{m, n}$. For $\lambda \in \mathbb{R}$ and $p \geq 1$, we set
\begin{equation}
    \Vert w_{m,n}\Vert_{\lambda, p} \coloneqq \left(\int_{B_1^{m+n}}\limits \frac{\dd k_1^m \dd \Tilde{k}_1^n}{\vert k_1^m\vert^{\lambda}\vert\Tilde{k}_1^n\vert^{\lambda}} \sup_{r \in [0,1]} \left\vert w_{m,n}[r; k_1^m; \Tilde{k}_1^n] \right\vert^p \right)^{1/p},
\end{equation}
in case that $p < \infty$, and
\begin{equation}
    \Vert w_{m,n}\Vert_{\lambda, \infty} \coloneqq \underset{(k_1^m, \Tilde{k}_1^n)\in B_1^{m+n}}{\text{ess sup}} \frac{\sup_{r\in [0,1]}\vert w_{m,n}[r;k_1^{m}; \Tilde{k}_1^{n}]\vert}{\vert k_1^{m}\vert^\lambda \vert \Tilde{k}_1^{n}\vert^\lambda},
\end{equation}
for $p = \infty$. In the case $m=n=0$, we set
\begin{align}
    \Vert w_{0,0}\Vert_{\lambda, p} \coloneqq& \sup_{r \in [0,1]} \vert w_{0,0}[r]\vert
\end{align}
and
\begin{equation}
    W_{0,0}[w_{0,0}] = w_{0,0}[H_{ph}],
\end{equation}
defined as an operator by functional calculus. Introducing a weight function $D: \mathbb{N}\cup \{0\} \mapsto \mathbb{R}_{>0}$, we define
\begin{equation}
    \mathcal{W}^D_{\lambda, p} \coloneqq \bigoplus_{m,n \geq 0} \mathcal{W}_{m, n}
\end{equation}
to be the Banach-space of sequences $\underline{w} = (w_{m,n})_{m,n \geq 0} \in \mathcal{W}^D_{\lambda, p}$ for which
\begin{equation}
    \Vert \underline{w}\Vert^D_{\lambda, p} \coloneqq \sum_{m,n\geq 0} D(m) D(n) \Vert w_{m,n}\Vert_{\lambda,p} < \infty \label{eq: DlambdapNormDef}.
\end{equation}
It is worth mentioning, that our results are completely independent of the regularity beyond continuity of the functions $w_{m,n}$ with respect to their first argument. For simplicity, we do not consider derivatives $\partial_r w_{m,n}$ of $w_{m,n}$ and their norms. Including the first derivative and choosing $p=2$, $\lambda = 3+2\mu$, for $\mu>0$, and $D(m) = \xi^{-m}$, for $0<\xi<1$, would lead to the Banach space $\mathcal{W}^\#_{\geq 0}$ with norm $\Vert \cdot\Vert^\#_{\mu, \xi}$ used in \cite{sfb}. Here an operator-theoretic RG was formulated, based on the isospectral smooth Feshbach map. This maps the operator \eqref{eq: specproblem} onto another bounded operator with the same spectrum at the spectral parameter $z = 0$ by reducing the degrees of freedom. To see that an iterative application of the renormalization transformation converges to a trivial Hamiltonian, the effect of the RG was studied on the Banach space $\mathcal{W}^\#_{\geq 0}$. In order to have a connection between $\mathcal{W}^\#_{\geq 0}$ and $\mathcal{B}(\mathcal{H}_{red})$ it was proven, that
\begin{equation}
        H : \mathcal{W}^\#_{\geq 0} \mapsto \mathcal{B}(\mathcal{H}_{red}) \quad, \quad \underline{w}\mapsto H[\underline{w}] = \sum_{m,n \geq 0} W_{m,n}[w_{m,n}] \label{eq: embeddingdef}
\end{equation}
defines a continuous and injective embedding, where the boundedness of $H \in \mathcal{B}(\mathcal{W}^\#_{\geq 0}; \mathcal{B}(\mathcal{H}_{red}))$ is necessary to make predictions about the spectrum of \eqref{eq: specproblem}. Moreover, this gives an unambiguous meaning to the operators \eqref{eq: WmnDef} and \eqref{eq: embeddingdef}.\\
We point out that the renormalization transformation, viewed as a map of bounded operators, is non-linear. This non-linearity is transcribed to the elements of the Banach space $\mathcal{W}^D_{\lambda, p}$ by Wick-ordering. The combinatorial factors arising from Wick-ordering products of creation and annihilation operators increase so fast, that it is in general hard to control them without a minimal stepsize in the scale reduction. In fact, we prove, that this is impossible in the sense of Theorem \ref{thm: maintheorem}. 
In contrast to the discrete operator-theoretic RG which was formulated and used in \cite{haslerherbst, RGAnalysis, sfb} to prove important results about the existence of ground states in models of non-relativistic QED, where a submultiplicative norm was not necessary, this seems to become false, when switching to a continuous RG.


\noindent As presented in \cite{salmhofer}, RG equations naturally involve non-linearities and the explicit formal solutions in component form consist of the above-mentioned combinatorial factors. Therefore, adequate norm estimates are of key importance for the proof of existence and uniqueness of solutions of RG equations. All in all, this leads us to some properties which should be satisfied by an adequate norm $\Vert\cdot\Vert^D_{\lambda, p}$, summarized in the following hypothesis.

\begin{hypo}\label{hy: norms}
Let $\Vert\cdot\Vert^D_{\lambda, p}: \mathcal{W}^D_{\lambda, p}\mapsto \mathbb{R}_{\geq 0}$ be defined by \eqref{eq: DlambdapNormDef}.
\begin{enumerate}
    \item $\Vert \cdot \Vert_{\lambda, p}^D$ controls the operator norm, such that
    \begin{equation}
        \Vert H[\underline{w}] \Vert_{op} \leq C_{D, \lambda, p} \Vert \underline{w}\Vert^D_{\lambda, p} \label{eq: OpNormDominanz}
    \end{equation}
    for some finite constant $C_{D, \lambda, p}$, where
    \begin{equation}
        H[\underline{w}] = w_{0,0}[H_{ph}] + \sum_{m+n \geq 1} W_{m,n}[w_{m,n}].
    \end{equation}
    \item There is a well-defined product $\ast: \mathcal{W}^D_{\lambda, p} \times \mathcal{W}^D_{\lambda, p} \mapsto \mathcal{W}^D_{\lambda, p}$ such that
    \begin{equation}
        H[\underline{w}_1] H[\underline{w}_2] = H[\underline{w}_1\ast \underline{w}_2]. \label{eq: Algebrenhomomorphismus}
    \end{equation}
    \item $\Vert \cdot\Vert^D_{\lambda, p}$ is a submultiplicative norm w.r.t. the product $\ast$ from \eqref{eq: Algebrenhomomorphismus}, meaning that
    \begin{equation}
        \Vert (\underline{w}_1 \ast \underline{w}_2) \Vert^D_{\lambda, p} \leq C^\prime_{D, \lambda, p} \Vert \underline{w}_1\Vert^D_{\lambda, p} \cdot \Vert \underline{w}_2\Vert^D_{\lambda, p} \label{eq: QuasiSubmultipl}
    \end{equation}
    for a second finite constant $C^\prime_{D, \lambda, p}$.
\end{enumerate}
\end{hypo}

\noindent The main result of this work is that none of the proposed norms satisfy Hypothesis \ref{hy: norms}. The precise result is formulated in the following Theorem.

\begin{theorem}\label{thm: maintheorem}
   Let $\lambda \in \mathbb{R}, p \geq 1$ and $D: \mathbb{N}_0 \mapsto \mathbb{R}_{>0}$ satisfy either
    \begin{itemize}
        \item[(a)] $\liminf_{M \to \infty} \tfrac{D(M)}{D(M+1)} > 0$ or
        \item[(b)] $ \lim_{M \to \infty}\tfrac{D(M)}{D(M+1)} = 0$. 
    \end{itemize}
    Then $\Vert \cdot\Vert^D_{\lambda, p}$, defined as before, does not fulfill Hypothesis \ref{hy: norms}.
\end{theorem}

\noindent The proof of this theorem is carried out in several steps:
\begin{itemize}
    \item For $p \in [1, 2)$ there are elements $w_{m,n} \in \mathcal{W}_{m,n}, \Vert w_{m,n}\Vert_{\lambda, p} < \infty$ such that $W_{m,n}[w_{m,n}] \notin \mathcal{B}(\mathcal{H}_{red})$. Roughly speaking, this is a consequence of the fact that $\mathcal{H}_{red}$ is a space of square-integrable functions and the inclusion $L^2(B_1^{m+n}) \subsetneq L^p(B_1^{m+n})$ for $2>p$. The statement holds true independent of D (Section \ref{sec: Necconditions}).
    \item For $p \geq 2$ and $\lambda < d(1-\tfrac{p}{2}) + \tfrac{p}{2}$ the product $\ast$ cannot be defined (Section \ref{sec: Necconditions}). Here, the integral kernel of the product \eqref{eq: Algebrenhomomorphismus} depends on an integral over the two kernels of the factors (see \eqref{eq: w3MN}). The condition $\lambda \geq d(1-\tfrac{p}{2}) + \tfrac{p}{2}$ turns out to be necessary for the existence of this integral.
    \item For $p \geq 2$, $\lambda \geq d(1-\tfrac{p}{2}) + \tfrac{p}{2}$ and $\liminf_{M \to \infty} \tfrac{D(M)}{D(M+1)} > 0$, there is no finite constant $C^\prime_{D, \lambda, p}$ such that \eqref{eq: QuasiSubmultipl} holds (Theorem \ref{thm: keineSubmulti}, Theorem \ref{thm: punendlich}). We construct two elements $\underline{w}(\epsilon), \underline{v}(\epsilon) \in \mathcal{W}^D_{\lambda, p}$ depending on a parameter $\epsilon >0$ such that $\Vert \underline{w}(\epsilon)\Vert^D_{\lambda, p} = \Vert \underline{v}(\epsilon)\Vert^D_{\lambda, p} = \mathrm{const}$, independently of $\epsilon$, and that $\lim_{\epsilon \downarrow 0} \Vert (\underline{w}(\epsilon)\ast \underline{v}(\epsilon)) \Vert^D_{\lambda, p} = \infty$. The growth of the combinatorial factors in the Wick-ordered product causes the divergence.
    \item For $p \geq 2$ and $\lim_{M \to \infty} \tfrac{D(M)}{D(M+1)} = 0$, there are elements $\underline{w}, \underline{v} \in \mathcal{W}^D_{\lambda, p}$ such that $\underline{w}\ast \underline{v} \notin \mathcal{W}^D_{\lambda, p}$, i.e., for which $\Vert\underline{w}\ast\underline{v}\Vert^D_{\lambda, p}=\infty$ (Theorem \ref{thm: fastgrowing}). The crucial point in the proof is that the $(M,N)$ component $(\underline{w}\ast \underline{v})_{M,N}$ of the Wick-ordered product depends on lower indices of the factors like $w_{M-m, n}$ (see equation \eqref{eq: w3MN}). Surprisingly, the divergence occurs even without the growth of the combinatorial factors.
\end{itemize}
To illustrate that these results are in fact non-trivial, we analyze some special examples in Section \ref{sec: specialcases}, where a submultiplicativity like \eqref{eq: QuasiSubmultipl} holds but not for the entire Banach-space $\mathcal{W}^D_{\lambda, p}$. These also illustrate the two major difficulties in estimating the norm of the Wick-ordered product (see Remark \ref{bem: twodifficulties}). Our interpretation of the result is that this absence of an adequate norm estimate may be the main reason why, in contrast to the first developed physical notion of a continuous RG by Wilson \cite{Wilson}, a rigorous mathematical formulation of an operator-theoretic RG is only given in a discrete form.

\section{Necessary conditions on the parameters of the norm} \label{sec: Necconditions}
For $p=2$ and $\lambda = 3+2\mu$ ($\mu >0$) it was proven in \cite{sfb}, that
\begin{equation}
    \Vert W_{m,n}[w_{m,n}] \Vert_{op} \leq \frac{\Vert w_{m,n}\Vert_{3+2\mu, 2}}{\sqrt{m^m n^n}}.
\end{equation}
In fact, we will see that $p \geq 2$ is necessary in order to have a bound like \eqref{eq: OpNormDominanz}. Let $p < 2, 0 < \rho \leq 1, m,m \geq 1$ and $\lambda \in\mathbb{R}$. First of all, we define the function
\begin{equation}
    w_{m,0}[k_1^{m}] = \left(\prod_{i=1}^m \frac{\vert k_i \vert^{\frac{\lambda}{p}}}{\vert \frac{\rho}{2} - \vert k_i\vert \vert^{\frac{2}{2+p}}} \right) \mathbf{1}[\Sigma [k_1^{m}] \leq \rho],
\end{equation}
which is an element of $\mathcal{W}^D_{\lambda, p}$ (independently of $D$), because
\begin{equation}
\begin{split}
        \left(\Vert w_{m, 0}\Vert_{\lambda, p}\right)^p &= \int_{B_1^m}\limits \dd k_1^m \left(\prod_{i=1}^m \frac{1}{\vert \frac{\rho}{2} - \vert k_i\vert \vert^{\frac{2p}{2+p}}} \right) \mathbf{1}[\Sigma [k_1^{m}] \leq \rho]\\
        &= (O_{d-1})^m \int_{[0,1]^m} \limits \left(\prod_{i=1}^m \frac{\dd r_i\; r_i^{d-1}}{\vert \frac{\rho}{2} -  r_i \vert^{\frac{2p}{2+p}}} \right) \mathbf{1}[\sum_{j=1}^m r_j \leq \rho] \\
        &< \infty,
    \end{split}
\end{equation}
as $\tfrac{2p}{2+p} < \tfrac{2p}{p+p} = 1$. Here, $O_{d-1}$ means the volume of the sphere $S^{d-1}$. Therefore
\begin{equation}
    w_{m,n}[k_1^m; \Tilde{k}_1^n] = w_{m,0}[k_1^{m}] \cdot w_{n,0}[\Tilde{k}_1^{n}]
\end{equation}
is an element of $\mathcal{W}^D_{\lambda, p}$, too. Next, we prove that $W_{m,n}[w_{m,n}]$ does not define a bounded operator on $\mathcal{H}_{red}$. We choose
\begin{equation}
    \psi_m(k_1^m) = \left(\prod_{i=1}^m \frac{1}{\vert \frac{\rho}{2} - \vert k_i\vert \vert^{\frac{p}{2+p}}} \right) \mathbf{1}[\Sigma [k_1^{m}] \leq \rho],
\end{equation}
where $\psi_m$ satisfies the desired symmetry condition for bosons and the integrability condition, as
\begin{equation}
    \begin{split}
        \Vert \psi_m \Vert^2_{L^2(B_1^m)} =& \int_{B_1^m}\limits \dd k_1^m \left(\prod_{i=1}^m \frac{1}{\vert \frac{\rho}{2} - \vert k_i\vert \vert^{\frac{2p}{2+p}}} \right) \mathbf{1}[\Sigma [k_1^{m}] \leq \rho]\\
        =& (O_{d-1})^m \int_{[0,1]^m}\limits \left(\prod_{i=1}^m \frac{\dd r_i\; r_i^{d-1}}{\vert \frac{\rho}{2} -  r_i \vert^{\frac{2p}{2+p}}} \right) \mathbf{1}[\sum_{j=1}^m r_j \leq \rho]\\
        <& \infty.
    \end{split}
\end{equation}
Finally, we see that the matrix element
\begin{equation}
    \begin{split}
        \langle \psi_m &\vert W_{m,n}[w_{m,n}] \psi_n \rangle\\
        &= \sqrt{m!n!}\int_{B_1^{m+n}}\limits \frac{\dd k_1^m \dd \Tilde{k}_1^n}{\vert k_1^m\vert^{\frac{1}{2}} \vert \Tilde{k}_1^n\vert^{\frac{1}{2}}} \psi_{m}(k_1^{m}) w_{m,n}[k_1^{m}; \Tilde{k}_1^n] \psi_n(\Tilde{k}_1^{n}) \\
        &= \sqrt{m!n!}\left\{\int_{B_1^m}\limits \left(\prod_{i=1}^m \frac{\dd^dk_i \; \vert k_i\vert^{\frac{\lambda}{p}- \frac{1}{2}}}{\vert \frac{\rho}{2} - \vert k_i \vert \vert^{\frac{2+p}{2+p}}} \right)\mathbf{1}[\Sigma [k_1^{m}] \leq \rho] \right\} \cdot \left\{m \leftrightarrow n \right\},
    \end{split}
\end{equation}
does not exist. This shows, that $p<2$ violates our first condition. Therefore in the following, we will only deal with $p\geq 2$.Next, we analyse the composition $\ast$ from equation \eqref{eq: Algebrenhomomorphismus}. From \cite{sfb} we cite the following proposition:

\begin{prop} \label{prop: Wickordering2}
Let $\underline{w}^1, \underline{w}^2 \in \mathcal{W}^\#_{\geq 0}$ and $F \in C^1_0([0, 1))$. Then the equation
\begin{align}
    H[\underline{w}^1] F[H_{ph}] H[\underline{w}^2] = H[\underline{w}] \label{eq: introductionF}
\end{align}
implicitly defines a sequence of functions $\underline{w} = (w_{M,N})_{M+N\geq 0}$ given by
\begin{align}
\begin{split}
    w_{M,N}[r; k_1^M; \Tilde{k}_1^N] &\coloneqq (\underline{w}^1 \ast_F \underline{w}^2)^{sym}_{M,N}[r; k_1^M; \Tilde{k}_1^N]\\
    &\coloneqq \frac{1}{M! N!} \sum_{\pi \in S_M, \nu \in S_N} \Tilde{w}_{M,N}[r; k_{\pi(1)}, ..., k_{\pi(M)}; \Tilde{k}_{\nu(1)}, ..., \Tilde{k}_{\nu(N)}]\\
    \Tilde{w}_{M,N}[r; k_1^M; \Tilde{k}_1^N]&\coloneqq (\underline{w}^1 \ast_F \underline{w}^2)_{M,N}[r; k_1^M; \Tilde{k}_1^N]\\
    \coloneqq \sum_{m=0}^M \sum_{n=0}^N \sum_{q\geq 0}&\binom{m+q}{q} \binom{n+q}{q} q!\\
     \int_{B_1^q}\limits \frac{\dd x^q_1}{\vert x_1^q\vert} &w^1_{M-m, n+q}[r + \Sigma[k_1^m]; k_{m+1}^M ; \Tilde{k}_1^n, x_1^q]\\
     F[r + &\Sigma [k_1^m] + \Sigma [\Tilde{k}_1^n] + \Sigma [x_1^q]] w^2_{m+q, N-n}[r + \Sigma [\Tilde{k}_1^n]; k_1^m, x_1^q; \Tilde{k}_{n+1}^N]
\end{split}\label{eq: w3MN}.
\end{align}
\end{prop}
\begin{proof}
A proof can be found in \cite{RGAnalysis}. Here, we use for simplicity a formulation without vacuum expectation values. 
\end{proof}

\begin{bem}
The choice of a smooth function $F$ in \eqref{eq: introductionF} is necessary, in order to introduce the concept of Wick-ordering on the Banach space of Hamiltonians on the reduced Hilbert space due to the continuity of $\Tilde{w}_{M,N}$ with respect to $r$. It is also possible to formally introduce the concept of Wick-ordering on the entire bosonic Fock space. In this case, we extend the integration in \eqref{eq: w3MN} to $\mathbb{R}^{dq}$ and leave the multiplication with $F$ out. In general, the reduction to the reduced Hilbert space lowers the norm of $w_{M,N}$, so it should be getting clear, that it is sufficient to prove our statements for the reduced Hilbert space. Since we are not interested in the regularity of the function $w_{m,n}$, for simplicity we change the composition $\ast_F$ to $\ast_\rho$ by replacing $F[r + \Sigma [k_1^m] + \Sigma [\Tilde{k}_1^n] + \Sigma [x_1^q]]$ in \eqref{eq: w3MN} with $\mathbf{1}[\Sigma [k_1^m] + \Sigma [\Tilde{k}_1^n] + \Sigma [x_1^q] < \rho]$, where $0< \rho \leq 1$.
\end{bem}
\begin{bem}
By the triangle inequality we have
\begin{equation}
    \Vert (\underline{w}^1 \ast_F \underline{w}^2)^{sym}_{M,N}\Vert_{\lambda, p} \leq \Vert (\underline{w}^1 \ast_F \underline{w}^2)_{M,N}\Vert_{\lambda, p},
\end{equation}
which implies
\begin{equation}
    \Vert (\underline{w}^1 \ast_F \underline{w}^2)^{sym}\Vert_{\lambda, p}^D \leq \Vert (\underline{w}^1 \ast_F \underline{w}^2)\Vert_{\lambda, p}^D.
\end{equation}
Therefore, in order to prove norm estimates for the product, it is sufficient to prove them for the non-symmetrized Wick-ordered product. Similar, if we would like to prove the absence of such estimates, a proof for the symmetrized product implies a proof for the non-symmetrized product.
\end{bem}

\noindent We immediately see, that $\ast_\rho$ cannot be defined, if the integral on the right-hand side of $\eqref{eq: w3MN}$ does not exist. By some example, we will see that this requires $\lambda \geq d(1-\tfrac{p}{2}) + \tfrac{p}{2}$.\\
Let $p\geq 2, \lambda < d(1-\tfrac{p}{2}) + \tfrac{p}{2}, m,n\in \mathbb{N}$ and $\epsilon = \tfrac{d\left(1 - \tfrac{p}{2}\right) + \tfrac{p}{2}-\lambda}{2}>0$. Next, we choose elements
\begin{equation}
\begin{split}
    \underline{v} \equiv& v_{m,0}[k_1^{m}] = \vert k_1^{m} \vert^{\frac{\lambda - d + \epsilon}{p}} \mathbf{1}[\Sigma [k_1^{m}] \leq \tfrac{\rho}{2}]\\
    \underline{w} \equiv& w_{0,n}[\Tilde{k}_1^{n}] = \vert \Tilde{k}_1^{n}\vert^{\frac{\lambda - d + \epsilon}{p}} \mathbf{1}[\Sigma [\Tilde{k}_1^{n}]\leq \tfrac{\rho}{2}],
\end{split}
\end{equation}
of $\mathcal{W}^D_{\lambda, p}$ (independently of $D$). The $(m-1, n-1)$ component of the product is given by
\begin{equation}
    \begin{split}
        (\underline{w}\ast_\rho \underline{v})_{m-1, n-1} =& \sum_{i=0}^{m-1} \sum_{j=0}^{n-1} \sum_{q \geq 0} \binom{i+q}{q} \binom{j+q}{q} q!\\
         \int_{B_1^q}\limits \frac{\dd x_1^q}{\vert x_1^{q}\vert} &w_{m-1-i, j+q}[r + \Sigma[k_1^i]; k_{i+1}^{m-1} ; \Tilde{k}_1^j, x_1^q]\\
    \mathbf{1}[\Sigma [k_1^i]& + \Sigma [\Tilde{k}_1^j] + \Sigma [x_1^q]\leq \rho] v_{i+q, n-1-j}[r + \Sigma [\Tilde{k}_1^j]; k_1^i, x_1^q; \Tilde{k}_{j+1}^{n-1}]\\
    =& \binom{m-1+1}{1}\binom{n-1+1}{1}\\
     \int_{B_1}\limits \frac{\dd^dx}{\vert x\vert}& w_{0,m}[\Tilde{k}_1^{n-1}, x] \mathbf{1}\left[\vert x\vert +\Sigma [\Tilde{k}_1^{n-1}] + \Sigma [k_1^{m-1}] \leq \rho \right] v_{n,0}[k_1^{m-1}, x]\\
    =& m \cdot n \vert K^{(m-1, n-1)} \vert^{\frac{\lambda - d + \epsilon}{p}}\\
    \int_{B_1}\limits \dd^dx& \vert x\vert^{\frac{2\lambda - 2d + 2\epsilon - p}{p}} \mathbf{1}\left[\vert x\vert +\Sigma [\Tilde{k}_1^{n-1}] + \Sigma [k_1^{m-1}] \leq \rho \right]\\
    & \mathbf{1}\left[\vert x \vert + \Sigma [k_1^{m-1}] \leq \tfrac{\rho}{2}\right] \mathbf{1}\left[\vert x \vert + \Sigma [\Tilde{k}_1^{n-1}] \leq \tfrac{\rho}{2} \right],
    \end{split}
\end{equation}
where the last integral diverges, since $\tfrac{2\lambda - 2d + 2\epsilon - p}{p} = \tfrac{\lambda - d + \tfrac{p}{2}(1-d) - p}{p} < -d$. Therefore, in the following we will only consider $\lambda \geq d(1-\tfrac{p}{2}) + \tfrac{p}{2}$.

\section{Special examples of norm estimates for the product}\label{sec: specialcases}
Before we start proving, that there is no submultiplicativity on the entire Banach space $\mathcal{W}^D_{\lambda, p}$, we analyse some important subsets of interactions for which estimates like \eqref{eq: QuasiSubmultipl} hold, undermining that our results are in fact non-trivial. First, we look at the norm $p = 2, \lambda = 3 + 2\mu, D(M)=\xi^{-M}$ already introduced in \cite{sfb}. Here, it is possible to define the composition $\ast_\rho$ for every component $(\underline{w}\ast_\rho \underline{v})_{M,N}$.

\newpage
\begin{theorem}\label{thm: KomponentenweiseAbsch}
Let $p=2, \lambda = 3+2\mu$ and $D(M) = \xi^{-M}$ for $\mu > 0$ and $0 < \xi < 1$. If $\underline{w}, \underline{v} \in \mathcal{W}^D_{3+2\mu, 2}$, then for all $M,N \in \mathbb{N}_0$ we have $(\underline{w}\ast_\rho \underline{v})_{M,N} \in \mathcal{W}_{M,N}$ and
\begin{equation}
    \Vert (\underline{w}\ast_\rho \underline{v})_{M,N} \Vert_{3+2\mu, 2} \leq (2\xi)^{M+N} \rho^{2+ 2\mu} e^{4\xi^2} \Vert \underline{w}\Vert_{3+2\mu, 2}^D \Vert \underline{v} \Vert_{3+2\mu, 2}^D \label{eq: KomponentenweiseAbsch}.
\end{equation}
\end{theorem}

\begin{proof}
We use the Cauchy-Schwarz inequality to estimate the absolute value of $(\underline{w}\ast_\rho \underline{v})_{M,N}[r; k_1^M; \Tilde{k}_1^N]$ and get
\begin{align}
\begin{split}
    &\vert (\underline{w}\ast_\rho \underline{v})_{M,N}[r; k_1^M; \Tilde{k}_1^N]\vert\\
    &\leq \sum_{m=0}^M \sum_{n=0}^N \sum_{q\geq 0}\binom{m+q}{q} \binom{n+q}{q} q! \\
    &\quad \left(\int_{B_1^q}\limits \frac{\dd x_1^{q}}{\vert x_1^{q}\vert}\mathbf{1}[\Sigma [x_1^q] < \rho] \sup_{r\in I}\left\vert w_{M-m, n+q}[r + \Sigma[k_1^m]; k_{m+1}^M ; \Tilde{k}_1^n, x_1^q]\right\vert^2\right)^{1/2}\\
    &\quad \left(\int_{B_1^q}\limits \frac{\dd x_1^{q}}{\vert x_1^{q}\vert}\mathbf{1}[\Sigma [x_1^q] < \rho] \sup_{r\in I}\left\vert v_{m+q, N-n}[r + \Sigma [\Tilde{k}_1^n]; k_1^m, x_1^q; \Tilde{k}_{n+1}^N] \right\vert^2 \right)^{1/2}.
\end{split}
\end{align}
In both integrals we insert factors $\vert x_1^{q}\vert^{2+2\mu} /\vert x_1^{q}\vert^{2+2\mu}$. By using the inequality 
\begin{align}
    \mathbf{1}[\Sigma [x_1^q] < \rho]\, \vert x_1^{q}\vert^{2+2\mu} \leq \frac{\rho^{2+2\mu}}{q^{(2+2\mu)q}}
\end{align}
we get
\begin{align}
\begin{split}
    &\vert (\underline{w}\ast_\rho \underline{v})_{M,N}[r; K^{(M,N)}]\vert\\
    &\leq \sum_{m=0}^M \sum_{n=0}^N \sum_{q\geq 0}\binom{m+q}{q} \binom{n+q}{q} \frac{q! \rho^{2+2\mu}}{q^{(2+2\mu)q}} \\
    &\quad \left(\int_{B_1^q}\limits \frac{\dd x_1^{q}}{\vert x_1^{q}\vert^{3+2\mu}} \sup_{r\in I}\left\vert w_{M-m, n+p}[r ; k_{m+1}^M ; \Tilde{k}_1^n, x_1^q]\right\vert^2\right)^{1/2}\\
    &\quad \left(\int_{B_1^q}\limits \frac{\dd x_1^{q}}{\vert x_1^{q}\vert^{3+2\mu}} \sup_{r\in I}\left\vert v_{m+q, N-n}[r; k_1^m, x_1^q; \Tilde{k}_{n+1}^N] \right\vert^2 \right)^{1/2}.
\end{split}
\end{align}
Using the triangle inequality of the $\Vert \cdot \Vert_{3+2\mu, 2}$-norm we have
\begin{align}
\begin{split}
     \Vert (\underline{w}\ast_\rho \underline{v})_{M,N}\Vert_{3+2\mu, 2} \leq& \sum_{m=0}^M \sum_{n=0}^N \sum_{q\geq 0}\binom{m+q}{q} \binom{n+q}{q} \frac{q!\rho^{2+2\mu}}{q^{(2+2\mu)q}} \\
    & \Vert w_{M-m, n+p}\Vert_{3+2\mu, 2} \, \Vert v_{m+p, N-n}\Vert_{3+2\mu, 2}.
\end{split}
\end{align}
Using $q!/q^q \leq 1$ and $\binom{m+q}{q} \leq 2^{m+q}$ and inserting some powers of $\xi$ yields
\begin{align}
    \begin{split}
     \Vert &(\underline{w}\ast_\rho \underline{v})_{M,N}\Vert_{3+2\mu, 2}\\
     &\leq \xi^{M+N} \rho^{2+ 2\mu}\sum_{m=0}^M \sum_{n=0}^N \sum_{q\geq 0} 2^{m+n} \frac{(\xi^2\cdot 4)^q}{q^{q}} \xi^{-(M-m + n+q)}\Vert w_{M-m, n+q}\Vert_{3+2\mu, 2} \\
     &\qquad \qquad\qquad \qquad\qquad \qquad\xi^{-(m+q+N-n)}\Vert v_{m+q, N-n}\Vert_{3+2\mu, 2} \\
    &\leq (2\xi)^{M+N}\rho^{2+ 2\mu} \sup_{i,j \in \mathbb{N}_0} \left(\xi^{-(i+j)} \Vert v_{i,j}\Vert_{3+2\mu, 2} \right) \\
    &\quad \sum_{m=0}^M \sum_{n=0}^N \sum_{q\geq 0} \frac{(4\xi^2)^q}{q!} \xi^{-(M-m + n+q)}\Vert w_{M-m, n+q}\Vert_{3+2\mu, 2}.
    \end{split}
\end{align}
Finally, we substitute $m\p = M-m, n\p = n+q$ and estimate the supremum by the associated sum, such that
\begin{align}
    \begin{split}
        \Vert (\underline{w}\ast_\rho \underline{v})_{M,N}\Vert_{3+2\mu, 2} \leq& (2\xi)^{M+N}\rho^{2+ 2\mu} \sum_{i+j\geq 0} \xi^{-(i+j)} \Vert v_{i,j}\Vert_{\lambda, p} \\
        & \sum_{q\geq 0} \frac{(4\xi^2)^q}{q!} \sum_{m^\prime =0}^M \sum_{n^\prime =q}^{N+q}  \xi^{-(m^\prime + n^\prime)}\Vert w_{m^\prime, n^\prime}\Vert_{3+2\mu, 2} \\
        \leq& (2\xi)^{M+N} \rho^{2+ 2\mu} e^{4\xi^2} \Vert \underline{v}\Vert_{3+2\mu, 2}^D \Vert \underline{w} \Vert_{3+2\mu, 2}^D.
    \end{split}
\end{align}
\end{proof}

\noindent We already see that equation \eqref{eq: KomponentenweiseAbsch} does not suffice to show that $(\underline{w}\ast_\rho \underline{v}) \in \mathcal{W}^D_{\lambda, p}$. Nevertheless, if $\underline{w}$ and $\underline{v}$ have an maximal index $K$, such that $w_{m,n} = v_{m,n} \equiv 0$ if $m>K$ or $n>K$, than $\underline{w}\ast_\rho \underline{v}$ has the maximal Index $2K$. Therefore we can estimate
\begin{equation}
    \Vert (\underline{w}\ast_\rho \underline{v})\Vert^D_{3+2\mu, 2} \leq 2^{4K} \rho^{2+2\mu} e^{4\xi^2} \Vert \underline{w} \Vert_{3+2\mu, 2}^D \, \Vert \underline{v}\Vert_{3+2\mu, 2}^D
\end{equation}
and the product is defined. However, due to the exponential increase of the constant $2^{4K} \rho^{2+2\mu} e^{4\xi^2}$ we can't extend this result to the entire $\mathcal{W}^D_{\lambda, p}$. We will prove in the following sections, that this is indeed impossible. Despite that, it is worth mentioning that there is a subset of interactions, for which a submultiplicativity holds.

\begin{lemma}\label{lem: BeispielFakultaet}
Let $p=2, \lambda \geq 1$ and $D(M) = (M!)$. If $\underline{w} \in \mathcal{W}^D_{\lambda, 2}$ is of the form $w_{m,n} \equiv 0$ for $m\neq 0$ and if $\underline{v} \in \mathcal{W}^D_{\lambda, 2}$ is of the form $v_{m,n} \equiv 0$ for $n\neq 0$, then we have $(\underline{w}\ast_\rho \underline{v}) \in \mathcal{W}^D_{\lambda, 2}$ and
\begin{equation}
    \Vert (\underline{w}\ast_\rho \underline{v})\Vert^D_{\lambda, 2} \leq e^1 \cdot \Vert \underline{w}\Vert^D_{\lambda, 2} \, \Vert \underline{v}\Vert^D_{\lambda, 2}.
\end{equation}
\end{lemma}

\begin{proof}
Similar to the previous theorem we estimate the $(M,N)$ component of the product
\begin{equation}
    \begin{split}
        &\vert (\underline{w}\ast_\rho \underline{v})_{M,N}[r; k_1^M; \Tilde{k}_1^N]\vert\\
    &\leq \sum_{m=0}^M \sum_{n=0}^N \sum_{q\geq 0}\binom{m+q}{q} \binom{n+q}{q} q! \\
    &\quad  \left(\int_{B_1^q}\limits \frac{\dd x_1^{q}}{\vert x_1^{q}\vert}\mathbf{1}[\Sigma [x_1^q] < \rho] \sup_{r\in I}\left\vert w_{M-m, n+q}[r + \Sigma[k_1^m]; k_{m+1}^M ; \Tilde{k}_1^n, x_1^q]\right\vert^2\right)^{1/2}\\
    &\quad \left(\int_{B_1^q}\limits \frac{\dd x_1^{q}}{\vert x_1^{q}\vert}\mathbf{1}[\Sigma [x_1^q] < \rho] \sup_{r\in I}\left\vert v_{m+q, N-n}[r + \Sigma [\Tilde{k}_1^n]; k_1^m, x_1^q; \Tilde{k}_{n+1}^N] \right\vert^2 \right)^{1/2}.
    \end{split}
\end{equation}
By the assumption all terms in the sums over $m$ and $n$ vanish except for the $m=M, n=N$ case. Using $\lambda \geq 1$ we get
\begin{equation}
    \begin{split}
        &\vert (\underline{w}\ast_\rho \underline{v})_{M,N}[r; k_1^M; \Tilde{k}_1^N]\vert\\
    &\leq \sum_{q\geq 0}\binom{M+q}{q} \binom{N+q}{q} q! \\
    &\quad \left(\int_{B_1^q}\limits \frac{\dd x_1^{q}}{\vert x_1^{q}\vert^\lambda}\mathbf{1}[\Sigma [x_1^q] < \rho] \sup_{r\in I}\left\vert w_{0, N+q}[r + \Sigma[k_1^M]; \Tilde{k}_1^N, x_1^q]\right\vert^2\right)^{1/2}\\
    &\quad \left(\int_{B_1^q}\limits \frac{\dd x_1^{q}}{\vert x_1^{q}\vert^\lambda}\mathbf{1}[\Sigma [x_1^q] < \rho] \sup_{r\in I}\left\vert v_{M+q, 0}[r + \Sigma [\Tilde{k}_1^N]; k_1^M, x_1^q] \right\vert^2 \right)^{1/2}.
    \end{split}
\end{equation}
This implies the following estimate for the $\Vert \cdot\Vert_{\lambda, 2}$-norm
\begin{equation}
    \begin{split}
        \Vert (\underline{w}\ast_\rho \underline{v})_{M,N}\Vert_{\lambda, 2} \leq& \sum_{q \geq 0} \frac{(M+q)!}{M! q!} \frac{(N+q)!}{N! q!} q! \Vert w_{0, N+q}\Vert_{\lambda, 2} \, \Vert v_{M+q, 0}\Vert_{\lambda, 2}\\
        =& \sum_{q\geq 0} \frac{1}{q!} \frac{D(N+q)}{D(N)} \Vert w_{0, N+q}\Vert_{\lambda, 2} \frac{D(M+q)}{D(M)} \Vert v_{M+q, 0}\Vert_{\lambda, 2}.
    \end{split}
\end{equation}
Finally, the desired inequality follows from
\begin{equation}
    \begin{split}
        \Vert (\underline{w}\ast_\rho \underline{v})\Vert^D_{\lambda, 2} \leq& \sum_{M \geq 0} \sum_{N\geq 0} \sum_{q\geq 0} \frac{1}{q!} D(N+q) \Vert w_{0, N+q}\Vert_{\lambda, 2} D(M+q) \Vert v_{M+q, 0}\Vert_{\lambda, 2}\\
        =& \sum_{q\geq 0}\frac{1}{q!}\sum_{N^\prime \geq q}D(N^\prime) \Vert w_{0, N^\prime}\Vert_{\lambda, 2} \sum_{M^\prime \geq q} D(M^\prime)\Vert v_{M^\prime, 0}\Vert_{\lambda, 2}\\
        \leq& e^1 \Vert \underline{w}\Vert^D_{\lambda, 2} \, \Vert \underline{v}\Vert^D_{\lambda, 2}.
    \end{split}
\end{equation}
\end{proof}

\begin{bem}\label{bem: twodifficulties}
Lemma \ref{lem: BeispielFakultaet} plays an important role in understanding the composition $\ast$. By choosing the left/right factor to consist solely of annihilation/creation operators, the number of contractions of two momenta and their combinatorial factors are the only things that affect the norm of the Wick-ordered product. By Lemma \ref{lem: BeispielFakultaet} we see, that by choosing weight functions like $D(M) = M!$, neutralizing the combinatorial factors, a submultiplicativity is possible. Therefore it is important to understand, that there are \textbf{two} major difficulties in estimating the norm of the Wick-ordered product. One is given by the growth of the combinatorial factors in \eqref{eq: w3MN}. The other one is given by the fact, that the $(M,N)$ component of the product depends not only on indices greater than $M$ and $N$ but also on lower indices like $w_{M-m, \cdot}$ and $v_{\cdot, N-n}$ (see equation \eqref{eq: w3MN}). Therefore demanding that $D(M)$ increases like a factorial causes problems if $\underline{w}$ and $\underline{v}$ also consist of creation/annihilation operators. It is the combination of these two effects, that is responsible for the absence of a submultiplicativity.
\end{bem}

\section{No submultiplicativity for sufficiently slow increasing weight functions}\label{sec: slowincreasing}
Up to this point we have shown that $p\geq 2$ and $\lambda \geq d(1-\tfrac{p}{2}) + \tfrac{p}{2}$ are necessary conditions in order to construct a norm $\Vert \cdot\Vert^D_{\lambda, p}$ which satisfies equations \eqref{eq: OpNormDominanz} and \eqref{eq: Algebrenhomomorphismus}. In this section we prove that even in this case a submulitplicity like equation \eqref{eq: QuasiSubmultipl} is not possible if $\liminf_{M \to \infty}\tfrac{D(M)}{D(M+1)} > 0$. For this proof we need the value of a certain family of integrals.

\begin{lemma}\label{lem: GammaIntegrale}
Let $M, m \in \mathbb{N}$, $M \geq 1$, $m \leq M$, $\rho > 0$, $x > 0$ and $y \geq 0$. Moreover, let
\begin{align}
    A(M, m, \rho, x, y) = \int_{(\mathbb{R}_0^+)^M} \limits \left(\prod_{i = 1}^M \frac{\dd s_i}{s_i} s_i^x \right) \mathbf{1}\left[\sum_{i=1}^M\limits s_i \leq \rho \right] \left(\rho - \sum_{j=1}^m \limits s_j \right)^y,
\end{align}
than we have
\begin{align}
     A(M, m, \rho, x, y) = \frac{\rho^{Mx + y} \cdot \Gamma\big[x\big]^M \Gamma\big[(M-m)x + y + 1\big]}{\Gamma\big[Mx + y + 1\big] \Gamma\big[(M-m)x + 1\big]}, \label{eq: Gammaintegral}
\end{align}
where
\begin{equation}
    \Gamma[x] = \int_0^\infty\limits t^{x-1} e^{-t} \dd t
\end{equation}
denotes the Gamma function.
\end{lemma}

\noindent The proof can be found in the Appendix \ref{sec: proofgammaint}. Using Lemma \ref{lem: GammaIntegrale} we can prove the following Theorem.

\begin{theorem}\label{thm: keineSubmulti}
Let $2\leq p < \infty, \lambda \geq d(1-\tfrac{p}{2}) + \tfrac{p}{2}$ and $\liminf_{M \to \infty} \tfrac{D(M)}{D(M+1)} \geq C_1 > 0$. There is no constant $C_{D, \lambda, p} < \infty$ such that for all $\underline{w}, \underline{v} \in \mathcal{W}^D_{\lambda, p}$
\begin{equation}
    \Vert (\underline{w}\ast_\rho \underline{v})^{sym}\Vert^D_{\lambda, p} \leq C_{D, \lambda, p} \Vert \underline{w}\Vert^D_{\lambda, p} \, \Vert \underline{v}\Vert^D_{\lambda, p} \label{eq: Submulti2}.
\end{equation}
\end{theorem}

\begin{proof}
The strategy of the proof is to construct two elements $\underline{w}(\epsilon)$ and $v(\epsilon)$ depending on a parameter $\epsilon > 0$ such that the norms $\Vert \underline{w}(\epsilon)\Vert^D_{\lambda, p} = K_1$ and $\Vert \underline{v}(\epsilon)\Vert^D_{\lambda, p} = K_2$ are independent of $\epsilon$ while $\lim_{\epsilon \downarrow 0}\Vert (\underline{w}(\epsilon)\ast_{\rho} \underline{v}(\epsilon))\Vert_{\lambda, p}^D$ diverges.\\
For all $m,n \in \mathbb{N}$ we define
\begin{align}
    w_{m,n}[r; k_1^m; \Tilde{k}_1^n] =& \begin{cases}
    c_n \cdot \vert \Tilde{k}_1^{n}\vert^{\frac{\lambda - d + \epsilon}{p}} \mathbf{1}[\Sigma [\tilde{k}_1^{n}] < \frac{\rho}{2}]&\text{, if }m = 0 \land n\geq 1\\
    0 &\text{, else}
    \end{cases}\\
    v_{m,n}[r; k_1^m; \Tilde{k}_1^n] =& \, w_{n,m}[r; k_1^n; \Tilde{k}_1^m].
\end{align}
The constants $c_n$ will be chosen later. By using Lemma \ref{lem: GammaIntegrale} we can compute the $\Vert \cdot\Vert_{\lambda, p}$-norm of these functions
\begin{align}
\begin{split}
    \Vert w_{0, m}(\epsilon)\Vert_{\lambda, p} =& \Vert v_{m, 0}(\epsilon)\Vert_{\lambda, p} \\
    =& c_m \left(\int_{B_1^m}\limits \frac{\dd k_1^{m}}{\vert k_1^{m}\vert^{\lambda}} \vert k_1^{m}\vert^{\lambda - d + \epsilon} \mathbf{1}\left[\Sigma [k_1^{m}] < \tfrac{\rho}{2}\right] \right)^{1/p} \\
    =& c_m (O_{d-1})^{m/p} \left(\int_{(\mathbb{R}_0^+)^m}\limits \left(\prod_{i = 1}^m\limits \frac{\dd s_i}{s_i} s_i^{\epsilon} \right) \mathbf{1}\left[\sum_{i=1}^m \limits s_i < \frac{\rho}{2}\right] \right)^{1/p} \\
    =& c_m (O_{d-1})^{m/p} A(m, 0, \tfrac{\rho}{2}, \epsilon, 0) \\
    =& c_m \left(\frac{\rho}{2}\right)^{m\epsilon/p}\left[\frac{\left(O_{d-1} \Gamma[\epsilon]\right)^m}{\Gamma[m\epsilon + 1]} \right]^{1/p}.
\end{split}
\end{align}
Now, we set
\begin{align}
    c_m \coloneqq \left(\frac{\rho}{2}\right)^{-m\epsilon/p} \left[\frac{\left(O_{d-1} \Gamma[\epsilon]\right)^m}{\Gamma[m\epsilon + 1]} \right]^{-1/p} \cdot \frac{1}{D(m)m^\kappa},
\end{align}
where $1 < \kappa < 3/2$ is an arbitrary constant such that
\begin{align}
    \Vert \underline{w}(\epsilon)\Vert_{\lambda, p}^D =& \Vert \underline{v}(\epsilon)\Vert_{\lambda, p}^D = \sum_{m = 1}^\infty \limits D(m)D(0) \Vert v_{m,0}(\epsilon)\Vert_{\lambda, p} = \sum_{m = 1}^\infty \limits \frac{D(0)}{m^\kappa} < \infty.
\end{align}
Next, the $(M,N)$-component of the (non-symmetrized) Wick-ordered product is given by
\begin{align}
\begin{split}
    (\underline{w}(\epsilon)&\ast_{\rho} \underline{v}(\epsilon))_{M,N}[k_1^M; \Tilde{k}_1^N] \\
    =& \sum_{m = 0}^M\limits \sum_{n = 0}^N\limits \sum_{q\geq 0}\limits \binom{m+q}{q}\binom{n+q}{q} q! \\
    &\int_{B_1^q}\limits \frac{\dd x_1^{q}}{\vert x_1^{q}\vert} \mathbf{1}\left[\Sigma[x_1^q] + \Sigma[k_1^m] + \Sigma[\tilde{k}_1^n] < \rho \right]\\
    & \quad w_{M-m, q+n}[k_{m+1}^{M}; x_1^{q}, \Tilde{k}_1^{n}](\epsilon) v_{q+m, N-n}[x_1^{q}, k_1^{m}; \Tilde{k}_{n+1}^{N}](\epsilon) \\
    =& \sum_{q \geq 0} \binom{M+q}{q}\binom{N+q}{q} q! \\
    & \int_{B_1^q}\limits \frac{\dd x_1^{q}}{\vert x_1^{q}\vert} \mathbf{1}\left[\Sigma[x_1^q] + \Sigma[k_1^M] + \Sigma[\tilde{k}_1^N] < \rho \right]\\
    & \mathbf{1}\left[\Sigma[x_1^q] + \Sigma[k_1^M] < \tfrac{\rho}{2} \right] \mathbf{1}\left[\Sigma[x_1^q] + \Sigma[\tilde{k_1^N}] < \tfrac{\rho}{2} \right] \\
    & c_{M+q} c_{N+q} \vert x_1^{q}\vert^{\frac{\lambda - d + \epsilon}{p}} \cdot \vert \Tilde{k}_1^{N}\vert^{\frac{\lambda - d + \epsilon}{p}} \cdot \vert x_1^{q}\vert^{\frac{\lambda - d + \epsilon}{p}} \cdot \vert k_1^{M}\vert^{\frac{\lambda - d + \epsilon}{p}} \\
    =& \sum_{q \geq 0} \binom{M+q}{q}\binom{N+q}{q} q! c_{M+q} c_{N+q} \\
    & \mathbf{1}\left[\Sigma[k_1^M] < \tfrac{\rho}{2} \right] \vert k_1^{M}\vert^{\frac{\lambda - d + \epsilon}{p}} \mathbf{1}\left[\Sigma[\tilde{k}_1^N] < \tfrac{\rho}{2} \right] \vert \Tilde{k}_1^{N}\vert^{\frac{\lambda - d + \epsilon}{p}} \\
    & \int_{B_1^q}\limits \frac{\dd x_1^{q}}{\vert x_1^{q}\vert} \mathbf{1}\left[\Sigma[x_1^q] < \tfrac{\rho}{2} - \max\left\{\Sigma[k_1^M], \Sigma[\tilde{k}_1^N]\right\} \right] \vert x_1^{q}\vert^{2\frac{\lambda - d + \epsilon}{p}} \\
    =& \left\vert (\underline{w}(\epsilon)\ast_{\rho} \underline{v}(\epsilon))_{M,N}[k_1^M; \Tilde{k}_1^N] \right\vert.
\end{split}
\end{align}
This function is completely symmetric in $k_1^M$ and $\tilde{k}_1^N$, so the symmetrized Wick-ordered product will be exact the same. \newpage

\noindent Again by Lemma \ref{lem: GammaIntegrale} we compute the last integral and estimate the absolute value of the product by the $q=1$ term
\begin{align}
\begin{split}
    \bigg\vert (\underline{w}(\epsilon)&\ast_{\rho} \underline{v}(\epsilon))^{sym}_{M,N}[k_1^M; \Tilde{k}_1^N]\bigg\vert \\
    =& \sum_{q \geq 0} \binom{M+q}{q}\binom{N+q}{q} q! c_{M+q} c_{N+q} \\
    & \mathbf{1}\left[\Sigma[k_1^M] < \tfrac{\rho}{2} \right] \vert k_1^{M}\vert^{\frac{\lambda - d + \epsilon}{p}} \mathbf{1}\left[\Sigma[\tilde{k}_1^N] < \tfrac{\rho}{2} \right] \vert \Tilde{k}_1^{N}\vert^{\frac{\lambda - d + \epsilon}{p}} \cdot (O_{d-1})^q \\
    & \int_{(\mathbb{R}_0^+)^q}\limits \left(\prod_{i=1}^q\limits \frac{\dd r_i}{r_i}r_i^{\frac{2\lambda - 2d + 2\epsilon + (d-1)p}{p}} \right)\mathbf{1}\left[\sum_{i=1}^q \limits r_i < \frac{\rho}{2} - \max\left\{\Sigma[k_1^M], \Sigma[\tilde{k}_1^N]\right\} \right] \\
    \geq& \binom{M+1}{1}\binom{N+1}{1} 1! c_{M+1} c_{N+1} \\
    &\mathbf{1}\left[\Sigma[k_1^M] < \tfrac{\rho}{2} \right] \vert k_1^{M}\vert^{\frac{\lambda - d + \epsilon}{p}} \mathbf{1}\left[\Sigma[\tilde{k}_1^N] < \tfrac{\rho}{2} \right] \vert \Tilde{k}_1^{N}\vert^{\frac{\lambda - d + \epsilon}{p}} \cdot (O_{d-1}) \\
    & \frac{\Gamma[\frac{2\lambda - 2d + 2\epsilon + (d-1)p}{p}]}{\Gamma[\frac{2\lambda - 2d + 2\epsilon + (d-1)p}{p} + 1]}\left(\tfrac{\rho}{2} - \max\left\{\Sigma[k_1^M], \Sigma[\tilde{k}_1^N]\right\}\right)^{\frac{2\lambda - 2d + 2\epsilon + (d-1)p}{p}} \\
    \geq& \frac{O_{d-1} p}{2\lambda - 2d + 2\epsilon + (d-1)p} \Bigg\{(M+1) c_{M+1} \mathbf{1}\left[\Sigma[k_1^M] < \tfrac{\rho}{2} \right] \\
    &\vert k_1^{M}\vert^{\frac{\lambda - d + \epsilon}{p}} \left(\tfrac{\rho}{2} - \Sigma[k_1^M] \right)^{\frac{2\lambda - 2d + 2\epsilon + (d-1)p}{p}}\Bigg\} \left\{(M \leftrightarrow N) \right\}.
\end{split}
\end{align}
This implies an estimate for the $\Vert \cdot\Vert_{\lambda, p}$-norm
\begin{align}
\begin{split}
    \Vert (\underline{w}(\epsilon)&\ast_{\rho} \underline{v}(\epsilon))^{sym}_{M,N}\Vert_{\lambda, p} = \left(\int_{B_1^{M+N}}\limits \frac{\dd k_1^M \dd \tilde{k}_1^N}{\vert k_1^M\vert^{\lambda} \vert \tilde{k}_1^N\vert^\lambda} \bigg\vert (\underline{w}(\epsilon)\ast_{\rho} \underline{v}(\epsilon))_{M,N}[k_1^M; \Tilde{k}_1^N]\bigg\vert^p \right)^{1/p} \\
    \geq& \frac{O_{d-1} p}{2\lambda - 2d + 2\epsilon + (d-1)p} (M+1)c_{M+1} (N+1) c_{N+1} \\
    &\left(\int_{B_1^M}\limits \frac{\dd k_1^M}{\vert k_1^M\vert^{\lambda}}\mathbf{1}\left[\Sigma[k_1^M] < \tfrac{\rho}{2} \right] \vert k_1^M\vert^{\lambda - d + \epsilon} \left(\tfrac{\rho}{2} - \Sigma[k_1^M] \right)^{2\lambda - 2d + 2\epsilon + (d-1)p} \right)^{1/p} \\
    & \left((M\leftrightarrow N) \right)^{1/p}.
\end{split}
\end{align}
Once more, Lemma \ref{lem: GammaIntegrale} gives us the values of the integrals so we get
\begin{align}
\begin{split}
    &\Vert (\underline{w}(\epsilon)\ast_{\rho} \underline{v}(\epsilon))^{sym}_{M,N}\Vert_{\lambda, p} \geq \frac{O_{d-1} p}{2\lambda - 2d + 2\epsilon + (d-1)p} (M+1)c_{M+1} (N+1) c_{N+1} \\
    & \left((O_{d-1})^{M}\int_{(\mathbb{R}_0^+)^M}\limits \left(\prod_{i=1}^M\limits \frac{\dd s_i}{s_i} s_i^{\epsilon} \right)\mathbf{1}\left[\Sigma_{i=1}^M s_i < \tfrac{\rho}{2} \right] \left(\tfrac{\rho}{2} - \Sigma_{i=1}^M s_i\right)^{2\lambda - 2d + 2\epsilon + (d-1)p} \right)^{1/p} \\
    & \quad\left((M\leftrightarrow N) \right)^{1/p} \\
    &= \frac{O_{d-1} p}{2\lambda - 2d + 2\epsilon + (d-1)p} \Bigg\{(M+1)c_{M+1} \left( \frac{\rho}{2}\right)^{(M\epsilon + 2\lambda - 2d + 2\epsilon + (d-1)p)/p}\\
    &\qquad \left(\frac{(O_{d-1}\Gamma[\epsilon])^M \Gamma[2\lambda - 2d + 2\epsilon + (d-1)p +1]}{\Gamma[M\epsilon + 2\lambda - 2d + 2\epsilon + (d-1)p + 1]}\right)^{1/p} \Bigg\}\{(M \leftrightarrow N) \}.
\end{split}
\end{align}
Inserting the definition of the constant $c_{M+1}$ yields
\begin{align}
\begin{split}
    \Vert (\underline{w}(\epsilon)&\ast_{\rho} \underline{v}(\epsilon))^{sym}_{M,N}\Vert_{\lambda, p}\\
    \geq& \frac{O_{d-1} p}{2\lambda - 2d + 2\epsilon + (d-1)p} \Bigg\{\frac{(M+1)}{D(M+1)(M+1)^{\kappa}} \left(\frac{\rho}{2}\right)^{-(M+1)\epsilon/p} \\
    &\qquad \left[\frac{\left(O_{d-1} \Gamma[\epsilon]\right)^{M+1}}{\Gamma[(M+1)\epsilon + 1]} \right]^{-1/p} \left( \frac{\rho}{2}\right)^{(M\epsilon + 2\lambda - 2d + 2\epsilon + (d-1)p)/p}\\
    &\qquad \left(\frac{(O_{d-1}\Gamma[\epsilon])^M \Gamma[2\lambda - 2d + 2\epsilon + (d-1)p +1]}{\Gamma[M\epsilon + 2\lambda - 2d + 2\epsilon + (d-1)p + 1]}\right)^{1/p} \Bigg\} \\
    & \{(M\leftrightarrow N) \} \\
    =& \frac{(O_{d-1})^{1 - \frac{2}{p}}p\Gamma[2\lambda - 2d + 2\epsilon + (d-1)p + 1]^{2/p}\left(\tfrac{\rho}{2} \right)^{\frac{2\lambda - 2d + \epsilon + (d-1)p}{p}}}{D(M+1)D(N+1)(2\lambda - 2d + 2\epsilon + (d-1)p)\Gamma[\epsilon]^{2/p}} \\
    & \Bigg\{\frac{1}{(M+1)^{\kappa - 1}} \left(\frac{\Gamma[(M+1)\epsilon + 1]}{\Gamma[M \epsilon + 2\lambda - 2d + 2\epsilon + (d-1)p + 1]}\right)^{1/p} \Bigg\} \\
    & \{(M\leftrightarrow N) \}.
\end{split}
\end{align}
The assumption $\lambda \geq d(1-\tfrac{p}{2}) + \tfrac{p}{2}$ implies that $2\lambda - 2d + 2\epsilon + (d-1)p \geq 2\epsilon$. Therefore it exists a constant $C_2 >0$ such that for sufficiently small $\epsilon >0$ we have
\begin{equation}
    \frac{(O_{d-1})^{1 - \frac{2}{p}}p\Gamma[2\lambda - 2d +2\epsilon + (d-1)p + 1]^{2/p}\left(\tfrac{\rho}{2} \right)^{\frac{2\lambda - 2d + \epsilon + (d-1)p}{p}}}{(2\lambda - 2d + 2\epsilon + (d-1)p)} \geq C_2.
\end{equation}
On the one hand, the gamma function $\Gamma[x]$ is monotonically increasing for $x\geq 2$. On the other hand $1/2 \leq \Gamma[x] \leq 1$ for $x \in [1,2]$, so all in all we get
\begin{align}
    \forall\, x, y \in [1, \infty): x \leq y \Rightarrow \frac{1}{2} \Gamma[x] \leq \Gamma[y].
\end{align}
Applying this inequality to $\Gamma[(M+1)\epsilon + 1]$ and $\Gamma[(N+1)\epsilon + 1]$ respectively, we estimate
\begin{align}
\begin{split}
    \Vert (\underline{w}(\epsilon)&\ast_{\rho} \underline{v}(\epsilon))^{sym}_{M,N}\Vert_{\lambda, p} \\
    \geq& \frac{C_2 \left(\tfrac{1}{2} \right)^{2/p}}{D(M+1)D(N+1)\Gamma[\epsilon]^{2/p}}\\
    &\Bigg\{\frac{1}{(M+1)^{\kappa - 1}} \left(\frac{\Gamma[M\epsilon + 1]}{\Gamma[M \epsilon + 2\lambda - 2d + 2\epsilon + (d-1)p + 1]}\right)^{1/p} \Bigg\} \\
    &\{(M\leftrightarrow N) \}.
\end{split}
\end{align}
Up to this point nothing is said about the concrete realisation of the limit $\epsilon \downarrow 0$. Now, we choose $\epsilon \equiv \epsilon(a) \coloneqq 1/a$ for $a \in \mathbb{N}$ such that $\epsilon \downarrow 0$ for $a \to \infty$. By assumption there exists some $b \in \mathbb{N}$ for which $0\leq b-1\leq 2\lambda - 2d + (d-1)p < b$ holds. Therefore $2\lambda - 2d + 2\epsilon(a) + (d-1)p < b$ is also true for sufficiently large $a \in \mathbb{N}$ implying
\begin{align}
\begin{split}
    \Vert (\underline{w}(\epsilon(a))&\ast_{\rho} \underline{v}(\epsilon(a)))^{sym}_{M,N}\Vert_{\lambda, p} \\
    \geq& \frac{C_2 \left(\tfrac{1}{2} \right)^{2/p}}{D(M+1)D(N+1)\Gamma[1/a]^{2/p}} \Bigg\{\frac{1}{(M+1)^{\kappa - 1}} \left(\frac{\Gamma[\frac{M}{a} + 1]}{\Gamma[\frac{M}{a} + b + 1]}\right)^{1/p} \Bigg\} \\
    & \{(M\leftrightarrow N) \} \\
    \geq& \frac{C_2 \left(\tfrac{1}{2} \right)^{2/p}}{D(M+1)D(N+1)\Gamma[1/a]^{2/p}} \Bigg\{\frac{1}{(M+1)^{\kappa - 1}} \left(\frac{1}{\left(\frac{M}{a} + b\right)^b}\right)^{1/p} \Bigg\} \\
    & \{(M\leftrightarrow N) \}.
\end{split}
\end{align}
Since $\liminf_{M \to \infty}\tfrac{D(M)}{D(M+1)} = C_1 > 0$ and  $D(M)>0$ for all $M \in \mathbb{N}_0$, there exists a constant $\tilde{C}_1>0$ such that $\tfrac{D(M)}{D(M+1)}\geq \tilde{C}_1$ for all $M \in \mathbb{N}_0$. Finally we get an estimate for the $\Vert \cdot \Vert^D_{\lambda, p}$-norm
\begin{align}
\begin{split}
    \Vert \underline{w}(\epsilon(a))&\ast_{\rho} \underline{v}(\epsilon(a))^{sym}\Vert_{\lambda, p}^D = \sum_{M,N \geq 0} D(M) D(N) \Vert(\underline{w}(\epsilon(a))\ast_{\rho} \underline{v}(\epsilon(a)))^{sym}_{M,N}\Vert_{\lambda, p} \\
    \geq& \frac{C_2 \left(\tfrac{1}{2} \right)^{2/p}}{\Gamma[1/a]^{2/p}} \left\{\sum_{M=1}^\infty \limits \frac{D(M)}{D(M+1)(M+1)^{\kappa -1}\left(\frac{M}{a} + b \right)^{b/p}} \right\}^2 \\
    \geq& \frac{C_2 \tilde{C}_1^2\left(\tfrac{1}{2} \right)^{2/p}}{\Gamma[1/a]^{2/p}} \left\{a^{b/p} \sum_{M=1}^\infty \limits \frac{1}{(M+ ba)^{\kappa -1}\left(M + ba \right)^{b/p}} \right\}^2 \\
    \geq& \frac{C_2 \tilde{C}_1^2\left(\tfrac{1}{2} \right)^{2/p}}{\left(\frac{1}{a}\Gamma[1/a]\right)^{2/p}} \left\{a^{(b-1)/p} \sum_{M=ba}^\infty \limits \frac{1}{\left(M + ba \right)^{b/p + \kappa - 1}} \right\}^2.
\end{split}
\end{align}
If $b<p$, i.e. $\lambda < d + \tfrac{p}{2}(2-d)$, we already get the desired divergence for sufficiently small $\kappa > 1$. Now let $b\geq p$. Using $1/a \Gamma[1/a] = \Gamma[1+ 1/a] \leq 1$ and estimating the sum by its integral we get
\begin{align}
\begin{split}
    \Vert \underline{w}(\epsilon(a))&\ast_{\rho} \underline{v}(\epsilon(a))^{sym}\Vert_{\lambda, p}^D\\
    &\geq C_2 \tilde{C}_1^2 \left(\frac{1}{2}\right)^{2/p} \left\{a^{(b-1)/p} \int_{M=ba}^\infty \limits \frac{\dd M}{\left(M + ba \right)^{b/p + \kappa - 1}} \right\}^2 \\
    &= C_2 \tilde{C}_1^2 \left(\frac{1}{2}\right)^{2/p} \left\{a^{(b-1)/p} \frac{1}{\left(b/p + \kappa - 2 \right)(2ba)^{b/p + \kappa - 2}} \right\}^2 \\
    &= \frac{C_2 \tilde{C}_1^2 \left(\frac{1}{2}\right)^{2/p}}{\left(b/p + \kappa - 2 \right)^2 (2b)^{2b/p + 2\kappa - 4}}a^{2(2-\kappa - 1/p)} \to\infty,
\end{split}
\end{align}
as $a\to \infty$, since $\kappa < 3/2$ and $p\geq 2$.
\end{proof}

\noindent For $p = \infty$ the situation is even worse, as the following Theorem shows.

\begin{theorem}\label{thm: punendlich}
Let $p=\infty, \lambda \in \mathbb{R}$ and $\liminf_{M \to \infty} \tfrac{D(M)}{D(M+1)} \geq C_1 > 0$. There are elements $\underline{w}, \underline{v} \in \mathcal{W}^D_{\lambda, \infty}$, such that $(\underline{w}\ast_{\rho}\underline{v})^{sym}\notin \mathcal{W}^D_{\lambda, \infty}$, i.e.,
\begin{equation}
    \Vert (\underline{w}\ast_{\rho}\underline{v})^{sym}\Vert^D_{\lambda, \infty} = \infty .
\end{equation}
\end{theorem}

\begin{proof}
Similar to the previous theorem we define
\begin{align}
    w_{m,n}[r; k_1^m; \Tilde{k}_1^n] =& \begin{cases}
    \underbrace{\frac{1}{D(n)n^2}}_{\coloneqq c_n} \cdot \vert \Tilde{k}_1^{n}\vert^{\lambda} \mathbf{1}[\Sigma [\tilde{k}_1^{n}] < \frac{\rho}{2}]&\text{, if }m = 0 \land n\geq 1\\
    0 &\text{, else}
    \end{cases}\\
    v_{m,n}[r; k_1^m; \Tilde{k}_1^n] =& \, w_{n,m}[r; k_1^n; \Tilde{k}_1^m],
\end{align}
such that
\begin{equation}
    \begin{split}
        \Vert \underline{w}\Vert^D_{\lambda, \infty} =& \Vert \underline{v}\Vert^D_{\lambda, \infty} = \sum_{M\geq 1} D(0)D(M) \Vert w_{0,M}\Vert_{\lambda, \infty}\\
        =& D(0) \sum_{M\geq 1} \frac{1}{M^2} = D(0)\frac{\pi^2}{6} < \infty.
    \end{split}
\end{equation}
The $(M,N)$ component of the product is given by
\begin{equation}
    \begin{split}
    (\underline{w}&\ast_{\rho} \underline{v})_{M,N}[k_1^M; \Tilde{k}_1^N] \\
    =&\sum_{q \geq 0} \binom{M+q}{q}\binom{N+q}{q} q! c_{M+q} c_{N+q} \\
    & \mathbf{1}\left[\Sigma[k_1^M] < \tfrac{\rho}{2} \right] \vert k_1^{M}\vert^{\lambda} \vert \tilde{k}_1^N\vert^\lambda \mathbf{1}\left[\Sigma[\tilde{k}_1^N] < \tfrac{\rho}{2} \right] \\
    & \int_{B_1^q}\limits \frac{\dd x_1^{q}}{\vert x_1^{q}\vert} \mathbf{1}\left[\Sigma[x_1^q] < \tfrac{\rho}{2} - \max\left\{\Sigma[k_1^M], \Sigma[\tilde{k}_1^N]\right\} \right] \vert x_1^{q}\vert^{2\lambda } \\
    \geq& \binom{M+2}{2}\binom{N+2}{2} 2 c_{M+2} c_{N+2} (O_{d-1})^2 \\
    & \mathbf{1}\left[\Sigma[k_1^M] < \tfrac{\rho}{2} \right] \vert k_1^{M}\vert^{\lambda} \vert \tilde{k}_1^N\vert^\lambda \mathbf{1}\left[\Sigma[\tilde{k}_1^N] < \tfrac{\rho}{2} \right] \\
    & \int_{(\mathbb{R}_{0}^+)^2}\limits \left(\prod_{i=1}^2 \frac{\dd s_i}{s_i} s_i^{2\lambda + d-1} \right) \mathbf{1}\left[\sum_{i=1}^2 s_i < \frac{\rho}{2} - \max\left\{\Sigma[k_1^M], \Sigma[\tilde{k}_1^N]\right\} \right],
    \end{split}
\end{equation}
where the integral in the last line already diverges if $\lambda \leq \tfrac{1-d}{2}$ implying $(\underline{w}\ast_\rho \underline{v})\notin \mathcal{W}^D_{\lambda, \infty}$. Therefore let $\lambda > \tfrac{1-d}{2}$. Moreover, we see that this product is again completely symmetric in  $k_1^M$ and $\tilde{k}_1^N$, so it is already symmetrized. We estimate
\begin{equation}
    \begin{split}
        \Vert(\underline{w}&\ast_{\rho} \underline{v})^{sym}_{M,N}\Vert_{\lambda, \infty} \\
        \geq& \frac{M^2N^2 c_{M+2}c_{N+2}}{2} (O_{d-1})^2 \\
        & \sup_{(k_1^M, \tilde{k}_1^N) \in B_1^{M+N}} \Bigg\{\mathbf{1}\left[\Sigma[k_1^M] < \tfrac{\rho}{2} \right] \mathbf{1}\left[\Sigma[\tilde{k}_1^N] < \tfrac{\rho}{2} \right] \\
        &\quad \frac{\left(\frac{\rho}{2} -  \max\left\{\Sigma[k_1^M], \Sigma[\tilde{k}_1^N]\right\} \right)^{2(2\lambda + d- 1)}\Gamma[2\lambda + d- 1]^2}{\Gamma[2(2\lambda + d- 1) + 1]}\Bigg\}\\
        =& M^2N^2 c_{M+2}c_{N+2} \frac{(O_{d-1})^2\left(\frac{\rho}{2} \right)^{2(2\lambda + d- 1)} \Gamma[2(2\lambda + d- 1)]^2}{2\Gamma[2(2\lambda + d- 1) + 1]}\\
        =& \mathrm{const} \cdot M^2N^2 c_{M+2}c_{N+2}
    \end{split}
\end{equation}
As already mentioned in the proof of Theorem \ref{thm: keineSubmulti}, there is a constant $\tilde{C}_1 > 0$ such that $\tfrac{D(M)}{D(M+1)} \geq \tilde{C}_1$ for all $M \in \mathbb{N}_0$. This implies the divergence of the $\Vert \cdot \Vert^D_{\lambda, \infty}$-norm since
\begin{equation}
    \begin{split}
        \Vert(\underline{w}\ast_{\rho} \underline{v})^{sym}\Vert^D_{\lambda, \infty} \geq& C_2 \left(\sum_{M \geq 1} D(M)M^2 c_{M+2} \right)^2\\
        =& C_2 \left(\sum_{M \geq 1} \frac{D(M)M^2}{D(M+2)(M+2)^2} \right)^2\\
        \geq& C_2 \left(\sum_{M \geq 1} \frac{1}{9}\frac{D(M)}{D(M+1)}\frac{D(M+1)}{D(M+2)} \right)^2\\
        \geq& C_2 \left(\frac{1}{9} \sum_{M\geq 1} \Tilde{C}_1^2 \right)^2 = \infty.
    \end{split}
\end{equation}
\end{proof}

\section{No well-defined product for fast increasing weight functions} \label{sec: fastincreasing}
Regarding Lemma \ref{lem: BeispielFakultaet} and the counterexample in Theorem \ref{thm: keineSubmulti} one might think, that choosing fast increasing weight functions like $D(M) = M!$ would solve the problems. It is possible to show, however, that $\mathcal{W}^D_{\lambda, p}$ is not a closed algebra with respect to $\ast_\rho$, if we choose weight functions, such that $ \lim_{M \to \infty }\tfrac{D(M)}{D(M+1)} = 0$. The proof relies on the following observation.

\begin{lemma}\label{lem: Summe}
Let $(a_n)_{n\in\mathbb{N}} \in (\mathbb{R}_{>0})^{\mathbb{N}}$ be a sequence, such that
\begin{equation}
    \lim_{n \to \infty} \frac{a_n}{a_{n+1}} = 0.
\end{equation}
Then, the sequence
\begin{equation}
    b_n \coloneqq \frac{a_{2n}}{(a_n)^2 n^\kappa} \quad, \quad \forall n \in \mathbb{N}
\end{equation}
is unbounded for all $\kappa \geq 0$.
\end{lemma}

\noindent A proof is given in the Appendix \ref{sec: proofsumme}. Now, we come to the final step in the proof of the absence of submultiplicative norms for Wick-ordered Operator Products.

\begin{theorem} \label{thm: fastgrowing}
Let $p\geq 2, \lambda \in \mathbb{R}$ and $\lim_{M \to \infty}\tfrac{D(M)}{D(M+1)} = 0$. There are elements $\underline{w}, \underline{v} \in \mathcal{W}^D_{\lambda, p}$, such that $(\underline{w}\ast_\rho \underline{v})^{sym}\notin \mathcal{W}^D_{\lambda, p}$, i.e., $\Vert (\underline{w}\ast_\rho \underline{v})^{sym}\Vert^D_{\lambda, p} = \infty$.
\end{theorem}
\begin{bem}
This includes the case  $D(M)=M!$. Therefore by Lemma \ref{lem: BeispielFakultaet} we already know, that we cannot use the exact same counterexample as in Theorem \ref{thm: keineSubmulti}.
\end{bem}

\begin{proof}
For $p < \infty$ we define
\begin{equation}
    \Tilde{w}_{m,n}[k_1^m; \Tilde{k}_1^n] = \prod_{i=1}^m \vert k_i\vert^{\lambda/p} \mathbf{1}\left[\vert k_i\vert \leq \tfrac{\rho}{2(m+n)}\right] \prod_{j=1}^n \vert \Tilde{k}_j\vert^{\lambda/p} \mathbf{1}\left[\vert \Tilde{k}_j\vert \leq \tfrac{\rho}{2(m+n)}\right],
\end{equation}
whereas for $p = \infty$ we set 
\begin{equation}
    \Tilde{w}_{m,n}[k_1^m; \Tilde{k}_1^n] = \prod_{i=1}^m \vert k_i\vert^{\lambda} \mathbf{1}\left[\vert k_i\vert \leq \tfrac{\rho}{2(m+n)}\right] \prod_{j=1}^n \vert \Tilde{k}_j\vert^{\lambda} \mathbf{1}\left[\vert \Tilde{k}_j\vert \leq \tfrac{\rho}{2(m+n)}\right]
\end{equation}
for all $m,n \in \mathbb{N}$. In both cases, we have $\Vert \Tilde{w}_{m,n}\Vert_{\lambda, p} < \infty$. Therefore we are able to define
\begin{equation}
    w_{m,n}[k_1^m; \Tilde{k}_1^n] = c_{m,n} \Tilde{w}_{m,n}[k_1^m; \Tilde{k}_1^n]
\end{equation}
with $c_{m,n} = \left(\Vert \Tilde{w}_{m,n}\Vert_{\lambda, p} D(m) D(n) m^\kappa n^\kappa \right)^{-1}$ for all $p\geq 2$ and some $\kappa > 1$. In addition, we choose $\underline{v} = (v_{m,n})_{m,n \in \mathbb{N}} \coloneqq (w_{m,n})_{m,n \in \mathbb{N}}$. Clearly, both sequences are in $\mathcal{W}^D_{\lambda, p}$ because
\begin{equation}
    \begin{split}
        \Vert \underline{w}\Vert^D_{\lambda, p} = \Vert \underline{v}\Vert^D_{\lambda, p} = \sum_{M,N \in \mathbb{N}} D(M) D(N) \Vert v_{m,n}\Vert_{\lambda, p} = \left(\sum_{M=1}^\infty \frac{1}{M^\kappa}\right)^2 < \infty,
    \end{split}
\end{equation}
as $\kappa > 1$. Next, let $M, N \in \mathbb{N}$. We write down the (non-symmetrized) $(2M, 2N)$ component of the Wick-ordered product
\begin{equation}
    \begin{split}
        (\underline{w}\ast_\rho \underline{v})_{2M, 2N}[k_1^{2M}; \Tilde{k}_1^{2N}] =& \sum_{m=0}^{2M} \sum_{n=0}^{2N} \sum_{q \geq 0} \binom{m+q}{q} \binom{n+q}{q} q!\\
        \quad&\int_{B_1^q} \frac{\dd x_1^q}{\vert x_1^q\vert} \mathbf{1}\left[\Sigma[x_1^q] + \Sigma[k_1^m] + \Sigma[\Tilde{k}_1^n] < \rho \right]\\
        \quad\quad & w_{2M - m, q+n}[k_{m+1}^{2M}; x_1^q, \Tilde{k}_1^n] v_{m+q, 2N - n}[k_1^m, x_1^q; \Tilde{k}_{n+1}^{2N}].
    \end{split}
\end{equation}
Again, all summands are positive, so we can estimate the absolute value of the left hand side by the summand given by $m=M, n=N$ and $q=0$
\begin{equation}
    \begin{split}
        \left\vert(\underline{w}\ast_\rho \underline{v})_{2M, 2N}[k_1^{2M}; \Tilde{k}_1^{2N}]\right\vert &\geq w_{M, N}[k_{M+1}^{2M}; \Tilde{k}_1^N] v_{M, N}[k_1^M; \Tilde{k}_{N+1}^{2N}] \\
        &\quad\times \mathbf{1}\left[\Sigma[k_1^M] + \Sigma[\Tilde{k}_1^N] < \rho \right].
    \end{split}
\end{equation}
Inserting the definition of $w_{M,N}$ and $v_{M,N}$, we see that the multiplication with the characteristic function in the last line is redundant and therefore we have
\begin{equation}
    \begin{split}
        &\left\vert(\underline{w}\ast_\rho \underline{v})_{2M, 2N}[k_1^{2M}; \Tilde{k}_1^{2N}]\right\vert\\
        &\quad \geq c_{M,N}^2 \prod_{i=1}^{2M} \vert k_i\vert^{\lambda/p} \mathbf{1}\left[\vert k_i\vert \leq \tfrac{\rho}{2(M+N)}\right] \prod_{j=1}^{2N} \vert \Tilde{k}_j\vert^{\lambda/p} \mathbf{1}\left[\vert \Tilde{k}_j\vert \leq \tfrac{\rho}{2(M+N)}\right],
    \end{split}
\end{equation}
for $p< \infty$, and
\begin{equation}
    \begin{split}
        &\left\vert(\underline{w}\ast_\rho \underline{v})_{2M, 2N}[k_1^{2M}; \Tilde{k}_1^{2N}]\right\vert\\
        &\quad \geq c_{M,N}^2 \prod_{i=1}^{2M} \vert k_i\vert^{\lambda} \mathbf{1}\left[\vert k_i\vert \leq \tfrac{\rho}{2(M+N)}\right] \prod_{j=1}^{2N} \vert \Tilde{k}_j\vert^{\lambda} \mathbf{1}\left[\vert \Tilde{k}_j\vert \leq \tfrac{\rho}{2(M+N)}\right],
    \end{split}
\end{equation}
for $p = \infty$. In both cases, the right hand side is completely symmetric in $k_1^{2M}$ and $\Tilde{k}_1^{2N}$ and we estimate
\begin{equation}
    \left\vert(\underline{w}\ast_\rho \underline{v})^{sym}_{2M, 2N}[k_1^{2M}; \Tilde{k}_1^{2N}]\right\vert \geq w_{M, N}[k_{M+1}^{2M}; \Tilde{k}_1^N] v_{M, N}[k_1^M; \Tilde{k}_{N+1}^{2N}].
\end{equation}
Now, we estimate the $\Vert \cdot \Vert_{\lambda, p}$-norm and get
\begin{equation}
    \Vert (\underline{w}\ast_\rho \underline{v})^{sym}_{2M, 2N}\Vert_{\lambda, p} \geq \Vert w_{M, N}\Vert_{\lambda, p} \Vert v_{M, N} \Vert_{\lambda, p},
\end{equation}
which implies
\begin{equation}
\begin{split}
    \Vert (\underline{w}\ast_\rho \underline{v})^{sym}\Vert^D_{\lambda, p} \geq& \sum_{M=1}^\infty \sum_{N=1}^\infty D(2M) D(2N) \Vert (\underline{w}\ast_\rho \underline{v})^{sym}_{2M, 2N}\Vert_{\lambda, p}\\
    \geq& \sum_{M=1}^\infty \sum_{N=1}^\infty D(2M) D(2N) \Vert w_{M, N}\Vert_{\lambda, p} \Vert v_{M, N} \Vert_{\lambda, p} \\
    =& \left(\sum_{M=1}^\infty \frac{D(2M)}{D(M)^2 M^{2\kappa}}\right)^2.
\end{split}
\end{equation}
But the last term diverges due to Lemma \ref{lem: Summe}.
\end{proof}

\section*{Acknowledgments}
This work emerged from a master thesis project under the supervision of V. Bach whom I am grateful to. Furthermore, I thank M. Ballesteros for helpful discussions.

\appendix
\section{Proof of Lemma \ref{lem: GammaIntegrale}}\label{sec: proofgammaint}
\begin{lem}
Let $M, m \in \mathbb{N}$, $M \geq 1$, $m \leq M$, $\rho > 0$, $x > 0$ and $y \geq 0$. Moreover, let
\begin{align}
    A(M, m, \rho, x, y) = \int_{(\mathbb{R}_0^+)^M} \limits \left(\prod_{i = 1}^M \frac{\dd s_i}{s_i} s_i^x \right) \mathbf{1}\left[\sum_{i=1}^M\limits s_i \leq \rho \right] \left(\rho - \sum_{j=1}^m \limits s_j \right)^y,
\end{align}
then we have
\begin{align}
     A(M, m, \rho, x, y) = \frac{\rho^{Mx + y} \cdot \Gamma\big[x\big]^M \Gamma\big[(M-m)x + y + 1\big]}{\Gamma\big[Mx + y + 1\big] \Gamma\big[(M-m)x + 1\big]}. \label{eq: GammaintegralVzwei}
\end{align}
\end{lem}
\begin{proof}
First of all, we simplify the proof to the case, where $\rho =1$. For arbitrary $\rho >0$ we substitute $\rho \Tilde{s}_i = s_i$ for all $i \in \{1, ..., M \}$
\begin{align}
\begin{split}
    A(M, m, \rho, x, y) =& \int_{(\mathbb{R}_0^+)^M} \limits \left(\prod_{i = 1}^M \frac{\dd s_i}{s_i} s_i^x \right) \mathbf{1}\left[\sum_{i=1}^M\limits s_i \leq \rho \right] \left(\rho - \sum_{j=1}^m \limits s_j \right)^y \\
    =& \rho^{Mx + y} \int_{(\mathbb{R}_0^+)^M} \limits \left(\prod_{i = 1}^M \frac{\dd \tilde{s}_i}{\tilde{s}_i} \tilde{s}_i^x \right) \mathbf{1}\left[\sum_{i=1}^M\limits \tilde{s}_i \leq 1 \right] \left(1 - \sum_{j=1}^m \limits \tilde{s}_j \right)^y \\
    =& \rho^{Mx + y} A(M, m, 1, x, y).
\end{split}
\end{align}
Next, we do the integration of the $s_j$ for $j \in \{1, ..., m\}$ in order to get a recursive formula of $A(M,m, 1, x, y)$
\begin{align}
\begin{split}
    A(M, m, 1, x, y) =& \int_{(\mathbb{R}_0^+)^M} \limits \left(\prod_{i = 1}^M \frac{\dd s_i}{s_i} s_i^x \right) \mathbf{1}\left[\sum_{i=1}^M\limits s_i \leq 1 \right] \left(1 - \sum_{j=1}^m \limits s_j \right)^y \\
    =& \int_0^1 \dd s_m\, s_m^{x-1} \int_{(\mathbb{R}_0^+)^{M-1}} \limits \left(\prod_{\substack{i = 1\\ i\neq m}}^M \frac{\dd s_i}{s_i} s_i^x \right) \mathbf{1}\left[\sum_{\substack{i = 1\\ i\neq m}}^M\limits s_i \leq 1 - s_m \right]\\
    & \left(1 - s_m - \sum_{j=1}^{m-1} \limits s_j \right)^y.
\end{split}
\end{align}
For all $i \in \{1, ..., M\}\setminus \{m\}$ we substitute $(1-s_m)\tau_i = s_i$ such that
\begin{align}
\begin{split}
    A(M, m, 1, x, y) =& \int_0^1 \dd s_m\, s_m^{x-1} (1-s_m)^{(M-1)x + y} \\
     & \int_{(\mathbb{R}_0^+)^{M-1}} \limits \left(\prod_{\substack{i = 1\\ i\neq m}}^M \frac{\dd \tau_i}{\tau_i} \tau_i^x \right) \mathbf{1}\left[\sum_{\substack{i = 1\\ i\neq m}}^M\limits \tau_i \leq 1 \right] \left(1 - \sum_{j=1}^{m-1} \limits \tau_j \right)^y \\
     =& B[x, (M-1)x + y + 1]\cdot A(M-1, m-1, 1, x, y),
\end{split}
\end{align}
where
\begin{align}
    B[x, y] = \int_0^1 \limits \dd s\, s^{x-1} (1-s)^{y-1}
\end{align}
denotes the Beta function. An Iteration yields
\begin{align}
    A(M, m, 1, x, y) = \left(\prod_{i = 1}^m\limits B[x, (M-i)x + y + 1] \right) \cdot A(M-m, 0, 1, x, y), \label{eq: ersteIteration}
\end{align}
for $m < M$, and
\begin{align}
\begin{split}
    A(M, M, 1, x, y) =& \left(\prod_{i = 1}^{M-1}\limits B[x, (M-i)x + y + 1] \right) \cdot A(1, 1, 1, x, y) \\
    =& \left(\prod_{i = 1}^{M-1}\limits B[x, (M-i)x + y + 1] \right) \cdot \int_0^1 \limits \dd s \, s^{x-1} (1-s)^y \\
    =& \left(\prod_{i = 1}^{M}\limits B[x, (M-i)x + y + 1] \right),
\end{split}
\end{align}
otherwise. In the second case all Integrals are solved and by the identity
\begin{align}
    B[x,y] = \frac{\Gamma[x]\Gamma[y]}{\Gamma[x+y]} \label{eq: BetaGammarelation}
\end{align}
we get
\begin{align}
\begin{split}
    A(M, M, 1, x, y) =& \left(\prod_{i = 1}^{M}\limits \frac{\Gamma[x]\Gamma[(M-i)x + y + 1]}{\Gamma[(M+1-i)x + y + 1]} \right) \\
    =& \frac{\Gamma[x]^M \Gamma[y + 1]}{\Gamma[Mx + y + 1]},
\end{split}
\end{align}
in accordance with \eqref{eq: GammaintegralVzwei}. In the first case let $N \geq 1$, then
\begin{align}
\begin{split}
    A(N, 0, 1, x, y) =& \int_{(\mathbb{R}_0^+)^N} \limits \left(\prod_{i = 1}^N \frac{\dd s_i}{s_i} s_i^x \right) \mathbf{1}\left[\sum_{i=1}^N\limits s_i \leq 1 \right] \\
    =& \int_0^1 \dd s_N\, s_N^{x-1} \int_{(\mathbb{R}_0^+)^{N-1}} \limits \left(\prod_{i = 1}^{N-1} \frac{\dd s_i}{s_i} s_i^x \right) \mathbf{1}\left[\sum_{i=1}^{N-1}\limits s_i \leq 1 - s_N \right].
\end{split}
\end{align}
Once more, we substitute $(1-s_N) \tau_i = s_i$ for all $i \in \{1, ..., N-1\}$ such that
\begin{align}
\begin{split}
    A(N, 0, 1, x, y) =& \int_0^1 \dd s_N\, s_N^{x-1} (1- s_N)^{(N-1)x} \\
    & \qquad \int_{(\mathbb{R}_0^+)^{N-1}} \limits \left(\prod_{i = 1}^{N-1} \frac{\dd \tau_i}{\tau_i} \tau_i^x \right) \mathbf{1}\left[\sum_{i=1}^{N-1}\limits \tau_i \leq 1\right] \\
    =& B[x, (N-1)x + 1] A(N-1, 0, 1, x, y).
\end{split}
\end{align}
An iteration yields
\begin{align}
\begin{split}
    A(N, 0, 1, x, y) =& \left(\prod_{i = 1}^{N-1}\limits B[x, (N-i)x + 1] \right) A(1, 0, 1, x, y) \\
    =& \left(\prod_{i = 1}^{N-1}\limits B[x, (N-i)x + 1] \right) \int_0^1\limits \dd s\, s^{x-1} \\
    =& \left(\prod_{i = 1}^{N-1}\limits B[x, (N-i)x + 1] \right) \cdot \frac{1}{x}.
\end{split}
\end{align}
Inserting this relation into equation \eqref{eq: ersteIteration} gives us
\begin{align}
\begin{split}
    A(M, m, 1, x, y) =& \left(\prod_{i = 1}^m\limits B[x, (M-i)x + y + 1] \right) \cdot A(M-m, 0, 1, x, y) \\
    =& \left(\prod_{i = 1}^m\limits B[x, (M-i)x + y + 1] \right) \cdot \frac{1}{x} \\
    & \left(\prod_{j = 1}^{M-m-1}\limits B[x, (M-m-j)x + 1] \right).
\end{split}
\end{align}
Finally, we simplify this expression by using equation \eqref{eq: BetaGammarelation} and $\Gamma[x+1] = x \Gamma[x]$, such that
\begin{align}
\begin{split}
    A(M, m, 1, x, y) =& \left(\prod_{i = 1}^m\limits \frac{\Gamma\big[x\big]\Gamma\big[(M-i)x + y + 1\big]}{\Gamma\big[(M+1 - i)x + y + 1\big]} \right) \cdot \frac{1}{x} \\
    & \left(\prod_{j = 1}^{M-m-1}\limits \frac{\Gamma\big[x\big]\Gamma\big[(M-m-j)x + 1\big]}{\Gamma\big[(M+1 - m - j)x + 1\big]} \right) \\
    =& \frac{\Gamma[x]^{M-1}}{x} \frac{\Gamma\big[(M-m)x + y + 1\big]}{\Gamma\big[Mx + y + 1\big]} \frac{\Gamma\big[x + 1\big]}{\Gamma\big[(M-m)x + 1\big]}\\
    =& \frac{\Gamma\big[x\big]^M \Gamma\big[(M-m)x + y + 1\big]}{\Gamma\big[Mx + y + 1\big] \Gamma\big[(M-m)x + 1\big]}.
\end{split}
\end{align}
\end{proof}

\section{Proof of Lemma \ref{lem: Summe}}\label{sec: proofsumme}
\begin{lem}
Let $(a_n)_{n\in\mathbb{N}} \in (\mathbb{R}_{>0})^{\mathbb{N}}$ be a sequence, such that
\begin{equation}
    \lim_{n \to \infty} \frac{a_n}{a_{n+1}} = 0.\label{eq: Nullfolge}
\end{equation}
Then, the sequence
\begin{equation}
    b_n \coloneqq \frac{a_{2n}}{(a_n)^2 n^\kappa} \quad, \quad \forall n \in \mathbb{N}
\end{equation}
is unbounded for all $\kappa \geq 0$.
\end{lem}

\begin{proof}
We prove this statement by contradiction. Let the assumption \eqref{eq: Nullfolge} be true and $b_n \leq K$ for all $n \in \mathbb{N}$ and some $K \in \mathbb{R}_{>0}$. In the first step, we prove by induction that
\begin{equation}
    a_{2^n} \leq \frac{1}{K} \left(a_1 K \right)^{2^n} 2^{\kappa (2^n - n - 1)}, \quad \forall n\in \mathbb{N}. \label{eq: a2nInduktion}
\end{equation}
For $n = 1$, this means
\begin{equation}
    a_2 \leq K (a_1)^2,
\end{equation}
which follows from $b_1 \leq K$. Now, let \eqref{eq: a2nInduktion} be true for $n-1$. Then, by the boundedness of $b_n$, we have
\begin{equation}
    a_{2^n} \leq K (a_{2^{n-1}})^2 (2^{n-1})^\kappa,
\end{equation}
and inserting the induction assumption for $a_{2^{n-1}}$ yields
\begin{equation}
    \begin{split}
        a_{2^n} &\leq K \left(\frac{1}{K}(a_1 K)^{2^{n-1}} 2^{\kappa(2^{n-1} - (n-1) - 1)} \right)^2 2^{\kappa (n-1)}\\
        &= \frac{1}{K} (a_1 K)^{2^n} 2^{\kappa(2^n - 2(n-1) - 2)} 2^{\kappa (n-1)}\\
        &= \frac{1}{K} (a_1 K)^{2^n} 2^{\kappa(2^n - n - 1)}.
    \end{split}
\end{equation}
Since $\kappa \geq 0$, we estimate
\begin{equation}
    a_{2^n} \leq \frac{1}{K}(2^\kappa a_1 K)^{2^n},
\end{equation}
which is a more convenient upper bound for $a_{2^n}$. In the second step, we see that for all $\epsilon > 0$, there exists an $N_\epsilon \in \mathbb{N}$ such that, for all $n \geq N_\epsilon$, we have
\begin{equation}
    \frac{a_n}{a_{n+1}} \leq \epsilon.
\end{equation}
Therefore, we choose $\epsilon > 0$ such that $2^\kappa a_1 K < 1/\epsilon$. Moreover, there exists a constant $L \in \mathbb{R}_{>0}$ such that
\begin{equation}
    \frac{a_n}{a_{n+1}} \leq L \quad, \quad \forall n\in\mathbb{N}.
\end{equation}
Choosing $n_{\epsilon} \in \mathbb{N}$ such that $2^{n_\epsilon} \geq N_\epsilon$ we get
\begin{equation}
\begin{split}
    a_{2^{n_\epsilon + m}} &\geq \frac{a_{2^{n_\epsilon + m} - 1}}{\epsilon} \geq \ldots \geq \left(\frac{1}{\epsilon}\right)^{2^{n_\epsilon + m} - 2^{n_\epsilon}} a_{2^{n_\epsilon}} \geq \left(\frac{1}{\epsilon}\right)^{(2^{m}-1) 2^{n_\epsilon}} \frac{a_{2^{n_\epsilon} - 1}}{L}\\
    &\geq \ldots \geq \left(\frac{1}{\epsilon}\right)^{(2^{m}-1) 2^{n_\epsilon}} \frac{a_1}{L^{2^{n_\epsilon} - 1}},
\end{split}
\end{equation}
for all $m \in \mathbb{N}$. Using the lower and the upper bound of $a_{2^{n_\epsilon + m}}$ we have
\begin{equation}
    \begin{split}
        \left(\frac{1}{\epsilon}\right)^{(2^{m}-1) 2^{n_\epsilon}} \frac{a_1}{L^{2^{n_\epsilon} - 1}} &\leq a_{2^{n_\epsilon + m}} \leq \frac{1}{K}(2^\kappa a_1 K)^{2^{n_\epsilon + m}}\\
        &= \frac{1}{K} (2^\kappa a_1 K)^{(2^{m} - 1)2^{n_\epsilon}} (2^\kappa a_1 K)^{2^{n_\epsilon}},
    \end{split}
\end{equation}
for all $m \in \mathbb{N}$. Finally, this yields a contradiction because
\begin{equation}
    \begin{split}
        0 < \frac{a_1 L K}{(L2^\kappa a_1 K)^{2^{n_\epsilon}}} \leq (\epsilon 2^\kappa a_1 K)^{(2^m - 1)2^{n_\epsilon}} \rightarrow 0,
    \end{split}
\end{equation}
for $m \to \infty$, as $\epsilon 2^\kappa a_1 K < 1$. Therefore, $b_n$ is unbounded.
\end{proof}

\bibliographystyle{unsrt}
\bibliography{Quellen}

\end{document}